\documentclass{article}
\usepackage[utf8]{inputenc}

\usepackage{macros}
\usepackage[linesnumbered,vlined]{algorithm2e}
\RestyleAlgo{ruled}
\SetKwComment{Comment}{/* }{ */}
\mathtoolsset{showonlyrefs}
\DontPrintSemicolon

\DeclareMathOperator*\cI{\mathcal{I}}
\DeclareMathOperator*\cS{\mathcal{S}}

\DeclareMathOperator*\cR{\mathcal{R}}
\DeclareMathOperator*\pred{\mathsf{pred}}
\DeclareMathOperator*\nxt{\mathsf{next}}
\DeclareMathOperator*\imax{end}
\DeclareMathOperator*\imin{start}
\newcommand{\getP}{\mathsf{GetPType}_w}
\newcommand\isQ{\mathsf{IsQType}_w}
\newcommand\getI{\mathsf{GetI}_w}

\newcommand\round{\text{round}}
\newcommand\Trivial{\mathsf{Trivial}}
\newcommand\FullLCS{\mathsf{FullLCSAlgorithm}}
\newcommand\Covering{\mathsf{CoveringAlgorithm}}
\newcommand\dedit{\vec\Delta_{\text{edit}}}
\newcommand\EqLCS{\mathsf{EqLCS}}
\newcommand\DP{\mathsf{DP}}

\title{Approximating binary longest common subsequence in almost-linear time}
\author{Xiaoyu He \thanks{Department of Mathematics, Princeton University, Princeton, NJ 08544, USA. Email: \url{xiaoyuh@princeton.edu}. Research supported by the NSF Mathematical Sciences Postdoctoral Research Fellowships Program under Grant DMS-2103154.} \and Ray Li \thanks{Department of Computer Science, UC Berkeley, Berkeley, CA, 94709, USA. Email: \url{rayyli@berkeley.edu}. Research supported by the NSF Mathematical Sciences Postdoctoral Research Fellowships Program under Grant DMS-2203067.}}
\date{\today}

\begin{document}

\maketitle

\begin{abstract}
  The Longest Common Subsequence (LCS) is a fundamental string similarity measure, and computing the LCS of two strings is a classic algorithms question. A textbook dynamic programming algorithm gives an exact algorithm in quadratic time, and this is essentially best possible under plausible fine-grained complexity assumptions, so a natural problem is to find faster approximation algorithms. When the inputs are two binary strings, there is a simple $\frac{1}{2}$-approximation in linear time: compute the longest common all-0s or all-1s subsequence. It has been open whether a better approximation is possible even in truly subquadratic time. Rubinstein and Song showed that the answer is yes under the assumption that the two input strings have equal lengths. We settle the question, generalizing their result to unequal length strings, proving that, for any $\varepsilon>0$, there exists $\delta>0$ and a $(\frac{1}{2}+\delta)$-approximation algorithm for binary LCS that runs in $n^{1+\varepsilon}$ time. As a consequence of our result and a result of Akmal and Vassilevska-Williams, for any $\varepsilon>0$, there exists a $(\frac{1}{q}+\delta)$-approximation for LCS over $q$-ary strings in $n^{1+\varepsilon}$ time.
  
  Our techniques build on the recent work of Guruswami, He, and Li who proved new bounds for error-correcting codes tolerating deletion errors. They prove a combinatorial ``structure lemma'' for strings which classifies them according to their oscillation patterns. We prove and use an algorithmic generalization of this structure lemma, which may be of independent interest.
\end{abstract}
%\tableofcontents

%%%%%%%%%%%%%%%%%%%%%%%% INTRO %%%%%%%%%%%%%%%%%%%%%%
\section{Introduction}
  In this paper, we give improved approximation algorithms for the Longest Common Subsequence (LCS), a fundamental string similarity measure that is of theoretical and practical interest.
  The LCS of two strings, as the name suggests, is the length of the longest sequence that appears as a (not necessarily contiguous) subsequence in both strings.
  The LCS is one of the most ubiquitous ways to quantify the similarity of two strings, a task that appears in a variety of contexts from spell checkers to DNA processing.

  Computing the LCS is a classic algorithms question. 
  A textbook dynamic programming algorithm gives an exact algorithm in quadratic time $O(n^2)$, while the fastest known algorithm runs in time $O(n^2/\log^2n)$ \cite{MasekP80}.
  Whether we can improve these algorithms has been a longstanding open question (see, for example, Problem 35 of \cite{Knuth72}).
  Under fine-grained complexity assumptions such as the Strong Exponential Time Hypothesis \cite{AbboudWW14, AbboudBW15, BringmannK15} and even more plausible hypotheses \cite{AbboudHWW16}, there is no exact algorithm for LCS in time $O(n^{2-\varepsilon})$ with $\varepsilon>0$.
  Because of these barriers for exact algorithms, it is natural to wonder whether there are faster approximation algorithms.

  When the inputs are two binary strings, the simple algorithm that computes the longest all-0s or all-1s common subsequence gives a $\frac{1}{2}$-approximation in linear time.
  Despite its simplicity, this has been the best known approximation for binary LCS on arbitrary inputs, even in truly subquadratic time ($n^{2-\varepsilon}$ for an absolute $\varepsilon>0$).
  This raises the following natural question.
  \begin{question}
    Do there exist $\delta,\eps >0 $ and a $(\frac{1}{2}+\delta)$-approximation algorithm of the LCS of two binary strings of length at most $n$ in time $O(n^{2-\varepsilon})$?
    \label{q:main}
  \end{question}
  Towards Question \ref{q:main}, Rubinstein and Song \cite{RubinsteinS20} showed that, if we assume the input strings have the same length, for all $\varepsilon>0$, there is a $(\frac{1}{2}+\delta)$-approximation of the LCS in time $O(n^{1+\varepsilon})$ ($\delta$ depends on $\varepsilon$).
  However, for the general setting of unequal length inputs remained open.

  Our main result answers Question \ref{q:main} in full, handling unequal length strings.
\newcommand\thmmain{
  For all $\eps>0$, there exists an absolute constant $\delta=\delta(\eps)>0$ and a deterministic algorithm that, given two binary strings $x$ and $y$ of not-necessarily-equal length, outputs a $(\frac{1}{2}+\delta)$-approximation of the longest common subsequence in time $O(n^{1+\eps})$ where $n=\max(|x|,|y|)$.}
\begin{theorem}
  \thmmain\footnote{Our runtime is actually $O(n\cdot \min(|x|,|y|)^{\varepsilon})$, which is slightly better in the case $y$ is much longer than $x$, but we state it as is for simplicity.}
\label{thm:main}
\end{theorem}
We note that our algorithm uses the equal-length LCS algorithm of \cite{RubinsteinS20} as a black box, so any improvements in the equal-length setting automatically yield improvements in the unequal-length setting.
In general, if there is an equal-length LCS algorithm running in time $T(n)$ giving a $(\frac{1}{2}+\delta)$-approximation, our algorithm gives a $O((n+T(n))\log^A n)$ time $(\frac{1}{2}+\delta^A)$-approximation on unequal length strings, for an absolute constant $A$.
Furthermore, while we present our algorithm as outputting the length of the longest common subsequence, we can output the subsequence of the promised length if the black-boxed equal-length LCS algorithm can.

  Our work gets around a technical barrier for unequal length strings, which was highlighted in \cite{AkmalV21}.
  The algorithms of \cite{RubinsteinS20} used the intimate connection between LCS and Edit Distance, the number of insertions, deletions, and substitutions needed to transform one string to another. 
  If we ignore substitutions, Edit Distance and LCS are in fact equivalent to compute exactly.
  Similar to LCS, there is no exact algorithm for Edit Distance in $n^{2-\varepsilon}$ time with $\varepsilon>0$, under plausible fine-grained complexity assumptions \cite{BackursI15, BringmannK15, AbboudHWW16}.
  Approximation algorithms for Edit Distance are well-studied, and a recent line of work \cite{BarYossefJKK04,AndoniKO10,AndoniO12,CDGKS20, GoldenbergRS20, Andoni18, BrakensiekR20, KouckyS20, AndoniN20} culminated in a constant-factor approximation of Edit Distance in almost-linear time \cite{AndoniN20}.
  Rubinstein and Song used these approximation algorithms for Edit Distance to obtain their approximation for LCS.
  However, because they rely on Edit Distance algorithms, they crucially use that the strings are equal length: note that if one string has length $n$ and the other string has length $100n$, their edit distance is always at least $99n$, so even computing a 3-approximation of edit-distance of the two input strings would be unhelpful for approximating LCS.

Our work gets around this problem by using different techniques to handle unequal length strings.
Our techniques are adapted from a recent work \cite{GuruswamiHL22} that proves lower bounds for error-correcting codes correcting deletions \cite{Levenshtein66,Ullman67} via the following combinatorial result about LCS.
\begin{theorem}[\cite{GuruswamiHL22}, deletion code limitation]
  There exists an absolute constants $A,\delta>0$ such that for any set $C$ of binary strings of length $n$ with $|C|\ge 2^{\log^A n}$, there exist two strings $x$ and $y$ with  $\LCS(x,y)\ge (\frac{1}{2}+\delta)n$.
\label{thm:ghl}
\end{theorem}
Intuitively, we may expect the techniques for Theorem~\ref{thm:ghl} to be useful here because it shares similarities with our main result, Theorem~\ref{thm:main}. While Theorem~\ref{thm:ghl} is a ``negative result'' for deletion codes, it is a ``positive result'' in the algorithmic sense, as it shows that among any small set of strings, two of them have a long common subsequence. Furthermore, it has a similar flavor as Question~\ref{q:main}, as both consider ``beating the trivial matching'' for LCS in binary strings. Thus, one might suspect then that these two problems are related, and we show indeed they are.
On the other hand, adapting the techniques from \cite{GuruswamiHL22} to our setting is nontrivial as we need to (i) make the combinatorial techniques algorithmic and (ii) handle unequal length strings (note in Theorem~\ref{thm:ghl} all strings are of the same length).

Computing LCS is also interesting over larger alphabets.
Approximating LCS when there is no restriction on the alphabet has been well studied \cite{HajiaghayiSSS19, RubinsteinSSS19, BringmannD21, AndoniNSS22, Nosatzki21}, and currently the best result \cite{AndoniNSS22, Nosatzki21} gives a randomized $\frac{1}{n^{o(1)}}$-approximation in linear time.
Over an alphabet of a given size $q>2$, there is, similar to the binary case, a trivial linear time $\frac{1}{q}$-approximation obtained by taking the longest common constant subsequence.
For fixed $q$, over general $q$-ary inputs, this was the best known approximation, even in subquadratic time. 
For $q$-ary inputs where the two strings have the same length, Akmal and Vassilevska-Williams \cite{AkmalV21} (see also \cite{Akmal21}) generalized the result of Rubinstein and Song, showing for all $\varepsilon>0$ there is a $(\frac{1}{q}+\delta)$-approximation in $n^{1+\varepsilon}$ time.

By the work of Akmal and Vassilevska-Williams \cite{AkmalV21}, our main result immediately implies improved approximation algorithms over non-binary alphabets, for the general case of not-necessarily-equal length strings.
Akmal and Vassilevska-Williams showed that if there is a $(\frac{1}{2}+\delta)$-approximation for binary LCS (which we show), there is a $(\frac{1}{q}+\delta')$-approximation for $q$-ary LCS in essentially the same runtime.
Hence, we have the following corollary.
\begin{corollary}[Follows from Theorem \ref{thm:main} and Theorem 1 of \cite{AkmalV21}]
  For all $\eps>0$ and integers $q\ge 2$, there exists an absolute constant $\delta=\delta(\eps)>0$ and a deterministic algorithm that, given two $q$-ary strings $x$ and $y$ of not-necessarily-equal length, outputs a $(\frac{1}{q}+\delta)$-approximation of the longest common subsequence in time $O(n^{1+\eps})$ where $n=\max(|x|,|y|)$.
\end{corollary}

%%%%%%%%%%%%%%% PRELIMINARIES %%%%%%%%%%%%%%%%%%%%%%%%%%%%%%%%%%%%%%%%%%

\section{Preliminaries}

For clarity of presentation, we sometimes drop floors and ceilings where they are not crucial.

\paragraph{Strings}
For a string $x$, a \emph{subsequence} of $x$ is any string obtained by deleting any number of bits of $x$. A \emph{substring} is a subsequence that appears as consecutive bits of $x$.
Let $0(x)$ and $1(x)$ denote the number of zeros and ones, respectively, in $x$.
A \emph{property} $P$ of binary strings is a set of binary strings. We say a string $x$ \emph{satisfies/has} a property $P$ if $x$ is in the set $P$.

\paragraph{Intervals}
We use interval notation similar to that of \cite{CDGKS20}.
By convention, an interval $I=[a,b]$ denotes the set $\{a+1,a+2,\dots,b\}$, and we write $\imin I=a$ and $\imax I = b$.
Note that $I$ and $I'$ are disjoint if and only if either $\imax I'\le \imin I$ or $\imax I\le \imin I'$.
The \emph{length} of an interval $I=[a,b]$ is $b-a$.
For a string $x$, let $x_I$ denote the contiguous substring $x_{a+1}x_{a+2}\cdots x_b$.
By abuse of notation, when the string $x$ is understood, we may use $I$ to refer to the substring $x_I$.
For an integer $w$, we say interval $I$ is \emph{$w$-aligned} if $\imin I$ and $\imax I$ are multiples of $w$.
An interval is a \emph{$w$-interval} if it has length $w$ and is $w$-aligned.
Let $\round_w(I)$ denote the largest $w$-aligned subinterval of $I$. Note we always have $|\round_w(I)| > |I|-2w$.

For an interval $I$ and integer $w$, let $\cI_{w}(I)$ be the collection of $w$-intervals that are subintervals of $I$.
When a string $x$ is understood (as it always will be), we write $\cI_{w}\defeq \cI_{w}(|x|)$.
Note that if $|x|$ is a multiple of $w$, the intervals of $\cI_w$ partition $[m]$.

A \emph{rectangle} is a product $I\times J$ where $I$ and $J$ are intervals.
A \emph{square} is a rectangle $I\times J$ with $|I|=|J|$.
A \emph{certified rectangle} is a pair $(I\times J, \kappa)$ where $\kappa$ is a positive number.

Define a partial ordering on intervals, where $I < I'$ iff $\imax I\le \imin I'$. That is, every element of $I$ is less than every element of $I'$. Note that if two intervals have nonempty intersection, they are incomparable.
We also define a partial ordering on rectangles, where $I\times J< I'\times J'$ iff $I<I'$ and $J<J'$.
We say a collection of (certified) rectangles is \emph{ordered} if any two (certified) rectangles are comparable under this partial order.

For any two strings $x$ and $y$, fix a canonical matching $\tau=\tau(x,y)$ between the bits of $x$ and $y$ that achieves the longest common subsequence ($\tau$ is not necessarily unique, but we can fix it to be, say, the lexicographically earliest one).
For $I\subset[|x|]$, let $J^\tau_I$ denote the (unique) smallest interval such that the bits of $x_I$ are only matched with bits in $y_{J^\tau_I}$ in the matching $\tau$.
Note that clearly if $I'$ and $I$ are disjoint, then $J^\tau_I$ and $J^\tau_{I'}$ are disjoint.

For any string $x$, we write $d(x)$ for the density of $x$, i.e. the ratio between the number of ones in $x$ and the length of $x$. For $\gamma>0$, we say an interval $I$ is \emph{$\gamma$-balanced in $x$} if $d(x_I) \in [\frac{1}{2}\pm \gamma]$, and we say $I$ is \emph{$\gamma$-imbalanced in $x$} otherwise.
If $x$ is understood (as it always will be), we simply say $\gamma$-balanced and $\gamma$-imbalanced.
A useful property of balanced strings $x$ is that we can find LCS close to $|x|/2$ with any other string of the same length.
\begin{lemma}
    If $x$ and $y$ are strings such that $x$ is $\gamma$-balanced and $|x|=|y|$, then $\LCS(x,y)\ge (\frac{1}{2}-\gamma)|x|$.
    \label{lem:balance}
\end{lemma}
\begin{proof}
    Suppose without loss of generality $y$ has at least $|y|/2$ ones. Then $y$ has at least $|x|/2$ ones. Since $x$ is $\gamma$-balanced, then $x$ has at least $(\frac{1}{2}-\gamma)|x|$ ones, so the LCS is at least $(\frac{1}{2}-\gamma)|x|$.
\end{proof}

\paragraph{Algorithms}
Let $\Trivial(x,y)$ denote the output of the simple algorithm that outputs the longest all-0s or all-1s subsequence. Clearly $\Trivial(x,y) = \max(\min(0(x),0(y)),\min(1(x),1(y))$.

Rubinstein and Song showed that one can obtain a $(\frac{1}{2}+\delta)$-approximation of equal length LCS. Their result immediately extends to a $(\frac{1}{2}+\delta')$-approximation for near-equal length LCS, which we use.
\begin{theorem}[Follows immediately from \cite{RubinsteinS20}]
  For any $\eps>0$, there exists a $\delta_{eq}=\delta_{eq}(\eps)>0$ and a $(\frac{1}{2}+\delta_{eq})$-approximation of the LCS of two binary strings $x$ and $y$ with $|x|,|y|\in[(1-\delta_{eq})n,(1+\delta_{eq})n]$ in time $O(n^{1+\eps})$.
\label{thm:rs20}
\end{theorem}
Let $\EqLCS(x,y)$ denote the output of the algorithm from Theorem~\ref{thm:rs20}.

%%%%%%%%%%%%%%%% PROOF SKETCH %%%%%%%%%%%%%%%%%%%%%%%%%%%%%%%%%%%%%
\section{Proof sketch}

In this section we give a high-level overview of our almost-linear time LCS approximation algorithm, Theorem~\ref{thm:main}. We start by describing the novel ingredient, our algorithmic structure lemma, Lemma~\ref{lem:struct}. It states, roughly speaking, that binary strings $s$ of length $w$ can be classified among one of $O(\log w)$ oscillation types or scales, such that for any two strings $s, t$ with the same type, there is a long subinterval $s_I$ in $s$ with $\LCS(s_I, t) > (1/2 + \delta) |s_I|$. Moreover, the lemma is algorithmic in that both the type of $s$ and the long subinterval $s_I$ are computable from $s$ in time nearly linear in $w$.

To formally define oscillation types, we first introduce the notion of a flag. In a string $x$, an \emph{$\ell$-flag} is an index $i$ such that between the $i$th one and the $(i+\ell)$-th one, there are strictly more than $10(\ell-1)$ zeros. In other words, an $\ell$-flag is a one-bit in $s$ that is immediately followed by a $0$-dense interval of length on the order of $\ell$. The existence of many $\ell$-flags in $x$ means that $x$ ``oscillates at scale $\ell$.''
An \emph{$\ell^+$-flag} is an index $i$ that is a $t$-flag for some $t\ge \ell$, where $t$ must be a power of two. The oscillation types guaranteed by Lemma~\ref{lem:struct} are as follows.

\begin{definition}[Definition~\ref{def:type} below]
  Let $\ell$ be a power of two, $\ell\in [1,w]$, and $x \in \{0,1\}^w$.
  \label{def:type-sketch}
  \begin{enumerate}
  \item We say that $x$ is \textit{$\ell$-coarse} if $\ell \ge \eps^2 w$ and there is a $\eps^2$-imbalanced interval $I$ in $x$ of length $\ell$. We say $x$ is \textit{coarse} if it is $\ell$-coarse for some $\ell \ge \eps^2 w$.

  \item We say that $x$ is \textit{$\ell$-fine} if it is not coarse, $\ell < \eps^2 w$, the number of $\ell^+$-flags in $x$ is at least $\eps w$, and $x$ contains $(0^{\ell}1^{\ell})^{\eps w/\ell}$ as a subsequence. We say $x$ is \textit{fine} if it is $\ell$-fine for some $\ell < \eps^2 w$.
  \end{enumerate}
\end{definition}

\noindent To a first approximation, this means that every string $x$ either has a long imbalanced subinterval or else behaves like the periodic string $(0^{\ell}1^{\ell})^{w/2\ell}$ for some $\ell$. 

Now we return to summarizing the proof of Theorem~\ref{thm:main}. By prior results \cite{RubinsteinS20,AkmalV21} (see also \cite{Akmal21}), it suffices to consider the ``perfectly balanced case,'' where the shorter string $x$ has an equal number of zeros and ones.
\newcommand\thmmaintech{
  For all $\varepsilon>0$, there exists an absolute constant $\delta=\delta(\varepsilon)>0$ and a deterministic algorithm that, on input strings $x$ and $y$ with $0(x)=1(x)\le \min(0(y),1(y))$, gives a $(\frac{1}{2}+\delta)$-approximation of the longest common subsequence in time $O(n^{1+\varepsilon})$.
}
\begin{theorem}
  \thmmaintech
\label{thm:main-2-sketch}
\end{theorem}
\noindent Theorem~\ref{thm:main} follows from Theorem~\ref{thm:main-2-sketch} by prior work \cite{RubinsteinS20,AkmalV21}; for completeness include the details in Section~\ref{app:imbalanced}.
In the rest of this section, we sketch the proof of Theorem~\ref{thm:main-2-sketch}.

Our algorithm for Theorem~\ref{thm:main-2-sketch}, which is described in pseudocode in Algorithms~\ref{alg:cover} and~\ref{alg:dp}, is a modification of the standard quadratric time DP algorithm for LCS, which we formulate as follows.
The standard DP algorithm computes $\LCS(x,y)$ by computing the full array $\DP[i][j] \defeq \LCS(x_{[i]}, y_{[j]})$, $0\le i \le |x|$, $0\le j\le |y|$ via the recursion
\[
\DP[i][j] = 
\begin{cases}
    \max\left(\DP[i][j-1], \DP[i-1][j],\DP[i-1][j-1] + 1\right) & \text{ if }x_i = y_i   \\
    \max\left(\DP[i][j-1], \DP[i-1][j]\right) & \text{ otherwise.}    
\end{cases}        
\]
In total this takes $O(|x|\cdot|y|)$ applications of the recursion. 
To prove Theorem~\ref{thm:main-2-sketch}, we don't need to compute the exact value of $\LCS(x,y)$, rather, we only need to output a value between $(1/2+\delta)\LCS(x,y)$ and $\LCS(x,y)$.
To estimate the LCS efficiently, we modify the DP above by recursing over large subrectangles instead of one bit at a time. 
We compute a collection of large rectangles $I\times J$ (where $I$ and $J$ are long subintervals of $[|x|]$ and $[|y|]$, respectively) and estimates $\kappa(I\times J)$ for their $\LCS$ (we call these \emph{certified rectangles}).
We guarantee that $\kappa(I\times J)\le \LCS(x_I,y_J)$ in every rectangle, and we desire that many of these $\kappa(I\times J)$ are good estimates of $\LCS(x_I,y_J)$.
(For readers familiar with \cite{CDGKS20}, finding these rectangles is analogous to their ``Covering Phase'').

The large rectangles under consideration are $\theta w$-aligned subrectangles of $[|x|]\times [|y|]$, where $w\approx |x|/\log |x|$ is a typical length of the $x$-intervals and $\theta$ is a small constant discretization parameter (in the real proof, we discretize $x$-intervals and $y$-intervals slightly differently, but ignore that here for sake of illustration).
Our modified DP algorithm is then
\begin{align}
\DP[i][j] =
    \max\big( &\DP\left[\imin I-\theta w\right]\left[\imin J\right], \DP\left[\imin I\right]\left[\imin J-\theta w\right],\nonumber\\
    &\DP\left[\imin I\right]\left[\imin J\right] + \kappa(I\times J)
    \text{ over $I\times J \in \cR$ s.t. $\imax I = i, \imax J = j$}\big),
\label{eq:dp}
\end{align}
where $\kappa(I\times J)$ denotes the lower bound for $\LCS(x_I, y_J)$ guaranteed by our certification algorithm. Observe that by induction $\DP[i][j]$ is still a lower bound for $\LCS(x_{[i]}, y_{[j]})$.
Because of our discretization, we only need to consider $i$ and $j$ a multiple of $\theta w$, so the number of dynamic programming states drops from $O(|x|\cdot|y|)$ to $O(\frac{|x|\cdot|y|}{(\theta w)^2}) \le \tilde O(\frac{|y|}{|x|})$.
Thus it remains to show we can quickly certify a collection of rectangles for which the dynamic program \eqref{eq:dp} outputs a $(\frac{1}{2}+\delta)$-approximation.

The main step is find a significant fraction of ``good'' rectangles for which $\kappa(I\times J) > (1/2+\gamma) \LCS(x_I,y_J)$.
We look for good rectangles in three different ways, as shown in Algorithm~\ref{alg:cover}. (1) First, we check for the ``trivial'' rectangles where $\Trivial(x_I, y_J) > (1/2+\gamma) |x_I|$ ($\gamma >0$ chosen very small). (2) Next, we black-box the equal-length LCS algorithm of Rubinstein and Song, and efficiently check for squares $I\times J$ with $|I|=|J|$ and $\LCS(x_I, y_J) > (1/2+\gamma) |x_I|$. (3) Finally --- and this is the main technical contribution of our work --- we use the algorithmic structure lemma, Lemma~\ref{lem:struct}, to efficiently compute ``oscillation frequencies'' for the intervals $I$ and $J$. For any given rectangle $I\times J$ where $|J|$ is longer than $|I|$, they can then be certified quickly by checking if $J$ has the same oscillation frequency as $I$. For technical reasons, for this last type of rectangle we are unable to guarantee that $\LCS(x_I, y_J) > (1/2+\gamma) |x_I|$ as we did for the other two types, but we can instead guarantee the weaker assumption that $I$ has a long subinterval $I'$ for which $\LCS(x_{I'}, y_J) > (1/2+\gamma) |x_{I'}|$. Handling this technicality requires some care, but to get across the main ideas we ignore this detail for the rest of this sketch and imagine that all certified rectangles satisfy $\LCS(x_I, y_J) > (1/2+\gamma) |x_I|$.

We also certify using the trivial algorithm a weaker set of ``default'' rectangles $I\times J$ where $\Trivial(I,J) \ge (1/2-\gamma^4)|I|$ (the constants are chosen for illustration). These rectangles exist all over the place and are used in the DP to fill in the gaps between the efficient ones above. 
We may assume $\LCS(x,y) \ge (1 - \delta)|x|$ --- or else the trivial matching, which is always $|x|/2$ by the setup of Theorem~\ref{thm:main-2-sketch}, gets a $(\frac{1+\delta}{2})$-approximation --- and this assumption guarantees we can certify many good rectangles and many default rectangles.
We show that while most of the $\kappa(I\times J)$ are $(1/2-\gamma^4)|I|$ coming from the default rectangles, a significant enough fraction of them are $(1/2+\gamma)|I|$ that for the final answer we get $\DP[|x|][|y|] > (1/2+\delta)|x|$. 

Since we are going for an almost-linear time algorithm (and not just subquadratic), we need to be careful to certify the rectangles quickly.
Note that, naively, there are roughly $(\frac{|y|}{\theta w})^2 \sim \tilde \Theta(\frac{|y|}{|x|})^2$ possible $J$-intervals.
If $y$ is much longer than $x$ (say $|y| = |x|^3$), then we cannot simply try to certify every rectangle, or else the runtime is super-linear in the input size, even if we can certify rectangles in constant time.
Instead, we restrict ourselves to certifying rectangles $I\times J$ where $J$ is ``minimal''. 
That is, for each $x$-interval $I$ and each $\imax(J)$, we look for the minimal $J$ where we can certify $\kappa(I\times J)\ge (1/2+\gamma)|I|$. 
We can find such $J$ by binary search (the ability to binary search depends on a technical property of the algorithmic structure lemma), so that the number of rectangles we are checking is now only $\tilde O(\frac{|y|}{|x|})$, rather than $\tilde O(\frac{|y|}{|x|})^2$.

%%%%%%%%%%%%%%%%%%%%%%%%%%%% STRUCTURE %%%%%%%%%%%%%%%%%%%%%%%%%%%%%%%%%%%

\section{Algorithmic Structure Lemma}
\label{sec:struct}

We now state and prove our algorithmic structure lemma.
We note that the final algorithm uses this lemma as a black-box, and can be understood without the proof of this lemma.
The interested reader can skip to Section~\ref{sec:alg} after Section~\ref{ssec:struct-statement}.

\subsection{Algorithm Structure Lemma Statement}
\label{ssec:struct-statement}

The following is the key technical lemma. 
It is inspired by and builds upon the ``Structure Lemma'' of \cite{GuruswamiHL22}, which was used to prove new deletion code bounds.

\begin{lemma}[Algorithm Structure Lemma]
  There exists an absolute constant $\delta_{code}>0$ such that for all sufficiently large $w$, there exists $T\le 2\log w$ and $2T$ string properties $P_1,\ldots, P_T, Q_1,\ldots, Q_T$ such that:
  \begin{enumerate}
  \item If string $x$ has length $w$, then there exists a $t\in [T]$ such that $x$ has property $P_t$.
  \item If $\LCS(x,y)\ge (1-\delta_{code})|x|$ and $x$ has property $P_t$, then $y$ (not necessarily of length $w$) has property $Q_t$.
  \item Property $Q_t$ is hereditary, meaning that if $y$ has $Q_t$ and $y$ is a subsequence of $y'$, then $y'$ has $Q_t$.
  \item For every $t\in[T]$, and strings $x$ and $y$, we can test if $x$ satisfies $P_t$ in time $O(w\log w)$.
  We can also preprocess the string $y$ in time $O(|y|\log|y|)$, such that we can answer queries of the form ``does $y'$ satisfy $Q_t$,'' for substrings $y'$ of $y$, in $O(w)$ time.
  \item If string $x$ has length $w$ and property $P_t$ and string $y$ has property $Q_t$, then there exists an interval $I\subset[w]$ such that $\LCS(x_I,y)\ge \frac{|I|}{2}+\delta_{code} w$.
  Furthermore, given $x$ and $t$, the interval $I$ and the promised common subsequence of $x_I$ and $y$ can be chosen independent of $y$, and both can be found in time $O(w\log w)$. 
  \end{enumerate}
\label{lem:struct}
\end{lemma}
\begin{remark}
  In item 5, it is easy to see that, if $\gamma\le \delta_{code}/10$, we may additionally assume (by starting with $\delta_{code}' = \delta_{code}/2$) that the interval $I$ is $\gamma w$-aligned by taking $I'\defeq \round_{\gamma w}(I)$.
  We do so in the application in Section~\ref{sec:alg}.
  \label{rem:struct}
\end{remark}

We now provide some more intuition for Lemma~\ref{lem:struct}.
First,  we describe the properties $P_t$ and $Q_t$ that we actually use (based on Definition~\ref{def:type-sketch}). For convenience, to define the properties, we index them as $P_{\ell,0}, P_{\ell, 1}, P_{\ell}$ for $\ell\le w$ a power of 2, for a total of roughly $T\sim 3\log w$ properties. 
\begin{itemize}
  \item If $\ell \ge \eps^2 w$ and $b\in\{0,1\}$, then $P_{\ell, b}$ is the property that $x$ is $\ell$-coarse, and its imbalanced interval of length $\ell$ is imbalanced in the direction of $b$-bits (i.e. has more $b$'s than $\bar{b}$'s).
  \item If $\ell \ge \eps^2 w$ and $b\in\{0,1\}$, then $Q_{\ell,b}$ is the property that $y$ has at least $\big(\frac{1+\eps^2}{2}\big)\ell$ $b$-bits.
  \item If $\ell < \eps^2 w$, then $P_\ell$ is the property that $x$ is $\ell$-fine.
  \item If $\ell < \eps^2 w$, then $Q_\ell$ is the property that $y$ contains $y_\ell = (0^{\ell}1^{\ell})^{\eps w/(5\ell)}$ as a subsequence.
  \end{itemize}

The properties $P_t$ are based on one of the key technical lemmas of the deletion codes bound \cite{GuruswamiHL22}, a combinatorial structure lemma.
This structure lemma roughly says that for strings of length $n$, there are properties $\tilde P_1,\dots, \tilde P_T$ with $T\le O(\log n)$, such that
\begin{enumerate}
\item[(i)] any binary string of length $n$ has some property $ \tilde P_i$, and
\item[(ii)] if two strings $x$ and $y$ have property $\tilde P_i$, then $x$ and $y$ have (continguous) substrings $x'$ and $y'$ of length $\Omega(n)$ with $\LCS(x',y')\ge (\frac{1}{2}+\delta)(\frac{|x'|+|y'|}{2})$ (the real guarantee is stronger but more technical to state).
\end{enumerate}
Theorem~\ref{thm:ghl}, the deletion codes lower bound, is proved (very roughly) by partitioning each string in $C$ into $\polylog n$ substrings, finding by pigeonhole two strings $x$ and $y$ such that the types of the corresponding substrings of $x$ and $y$ agree, and using guarantee (ii) to find a $(\frac{1}{2}+\delta')w$ overall LCS.

Lemma \ref{lem:struct} is a generalization of this combinatorial structure lemma that is ``algorithmic'' and ``handles unequal length strings.''
The properties $P_t$ that we choose in Lemma \ref{lem:struct} are similar to the properties $\tilde P_t$ of \cite{GuruswamiHL22}, and it is not hard to check by inspection that the properties $\tilde P_t$ of \cite{GuruswamiHL22} can be tested in linear time.
The difficulty lies in finding properties $Q_t$ of strings $y$ that (a) can be ``inherited'' from properties like $\tilde P_t$ if $y$ has a subsequence covering most of $x$, (b) can be tested efficiently, and (c) still guarantee an LCS advantage between $x$ and $y$.

Because of Lemma \ref{lem:struct}, we can define the following algorithms, which we use in our final LCS algorithm.
\begin{definition}
  For an integer $w$, let $P_1,\dots,P_T,Q_1,\dots,Q_T$ be the properties from Lemma~\ref{lem:struct}.
  Let $\getP(x)$ denote the smallest index $t$ such that $x$ has property $P_t$. 
  Let $\isQ(x,t)$ be true if $x$ has property $Q_t$ and false otherwise.
  For $x$ satisfying $P_t$ for some $t$, let $\getI(x,t)$ denote a $\gamma w$-aligned interval $I$ such that $\LCS(x_I,y)\ge \frac{|I|}{2}+\delta_{code}w$ for all $y$ satisfying $Q_t$.
  Such an interval exists by Lemma~\ref{lem:struct} and Remark\ref{rem:struct}.
\end{definition}
By Lemma~\ref{lem:struct}, $\getP(x)$ can be computed in $O(w\log^2 w)$ time, since one can simply test each of the $O(\log w)$ properties $P_t$.
By Lemma~\ref{lem:struct}, any string $y$ can be preprocessed in $O(|y|\log |y|)$ time such that, for any contiguous substring $y'$ of $y$, $\isQ(y',t)$ can be computed in $O(w)$ time.
Note that the input to $\getP$ must have length $w$, but the string input to $\isQ$ can have arbitrary length.
By Lemma~\ref{lem:struct}, $\getI(x,t)$ can be computed in $O(w\log w)$ time.

\subsection{Combinatorial structure lemma and types}

In a string $x$, an \emph{$\ell$-flag} is an index $i$ such that between the $i$th one and the $(i+\ell)$-th one, there are strictly more than $10(\ell-1)$ zeros.
An \emph{$\ell^+$-flag} is an index $i$ that is a $t$-flag for some $t\ge \ell$, where $t$ must be a power of two. 
By abuse of notation, if $i$ is a $\ell$-flag, we may also call the $i$-th one of $x$ a $\ell$-flag.
Note that there are many more zeros than ones between the $i$th and $(i+\ell)$-th one, so flags tell us where it is more advantageous to use zeros rather than ones in finding long subsequences.

The basis for the algorithmic structure lemma is a combinatorial structure lemma for strings, which we inherit from \cite[Lemma 4.1]{GuruswamiHL22}. We use a weaker form of the lemma, which has a significantly simpler statement and proof, and is also tailored to our algorithmic application. The proof is given in Appendix~\ref{apx:comb-struct}. 

\begin{lemma}[Combinatorial Structure Lemma]\label{lem:comb-struct}
  For $\eps = 10^{-5}$ and $w$ sufficiently large, at least one of the following two conditions holds for every string $x \in \{0,1\}^w$.
  \begin{enumerate}
  \item There exists $\ell \in [\eps^2 w, w]$ equal to a power of two and an $0.1$-imbalanced interval $I$ in $x$ of length $\ell$.

  \item There exists $\ell \in [1, \eps^2 w)$ equal to a power of two such that the number of $\ell^+$-flags in $x$ is at least $\eps w$, and $x$ contains $(0^{\ell}1^{\ell})^{\eps w/\ell}$ as a subsequence.
  \end{enumerate}
\end{lemma}

For the remainder of Section~\ref{sec:struct}, fix $\eps\defeq 10^{-5}.$ Lemma~\ref{lem:comb-struct} shows that every sufficiently long string $x$ is of one of the below types, of which there are $\log w$ total. 

\begin{definition}\label{def:type}
  Let $\ell$ be a power of two, $\ell\in [1,w]$, and $x \in \{0,1\}^w$.
  \begin{enumerate}
  \item We say that $x$ is \textit{$\ell$-coarse} if $\ell \ge \eps^2 w$ and there is a $\eps^2$-imbalanced interval $I$ in $x$ of length $\ell$. We say $x$ is \textit{coarse} if it is $\ell$-coarse for some $\ell \ge \eps^2 w$.

  \item We say that $x$ is \textit{$\ell$-fine} if it is not coarse, $\ell < \eps^2 w$, the number of $\ell^+$-flags in $x$ is at least $\eps w$, and $x$ contains $(0^{\ell}1^{\ell})^{\eps w/\ell}$ as a subsequence. We say $x$ is \textit{fine} if it is $\ell$-fine for some $\ell < \eps^2 w$.
  \end{enumerate}
\end{definition}

Note that for the convenience of our later proofs, we change the imbalanced threshold from $0.1$ in Lemma~\ref{lem:comb-struct} to the much smaller $\eps^2$ in the above definition. Since every $0.1$-imbalanced interval is also $\eps^2$-imbalanced, Lemma~\ref{lem:comb-struct} implies every sufficiently long string $x$ is of one of the above two types.

\subsection{Algorithmic structure lemma ingredients}
As in the last section, we fix $\eps = 10^{-5}$. Also, for brevity, for every positive integer $\ell$, define the special string
\[
  y_\ell \defeq (0^{\ell}1^{\ell})^{\eps w/(5\ell)}.
\]

We prove two ingredients to justify our ``$Q_t$'' properties in the algorithmic structure lemma.
The first is the simple observation that if $x$ is $\ell$-fine and $\LCS(x,y) \ge (1-\delta_{code})|x|$, then $y$ inherits the easily testable subsequence $y_\ell$ from $x$.
\begin{lemma}
  Let $0 < \delta < \eps/10$, $\ell$ be a power of two, and $w > \eps^{-2}\ell$.
  If $x$ is an $\ell$-fine string of length $w$ and $\LCS(x,y)\ge (1-\delta)w$, then $y_\ell$ is a subsequence of $y$.
  \label{lem:long-1}
\end{lemma}
\begin{proof}
  By definition, if $x$ is $\ell$-fine then $x$ contains $(0^{\ell}1^{\ell})^{\eps w/\ell} = y_{\ell}^5$ as a subsequence.
  Since $\LCS(x,y)\ge |x|-\delta w$ and $y_\ell^5$ is a subsequence of $x$, we have $\LCS(y_\ell^5,y)\ge |y_\ell^5|-\delta w$.
  Thus, there is a subsequence of $y$ obtained by applying $\delta w$ deletions to $y_\ell^5$.
  By counting, at most $2\delta w/\ell$ of the chunks $0^{\ell}1^\ell$ in $y_\ell^5$ receive more than $\ell/2$ of these deletions. 
  The remaining $\eps w/\ell-2\delta w/\ell > 4\eps w/(5\ell)$ chunks each have at least $\ell/2$ zeros and $\ell/2$ ones. Taking $\ell$ zeros from the first two chunks, $\ell$ ones from the next two, and so on, we see that $y$ contains $y_\ell$ as a subsequence, as desired.
\end{proof}

The second ingredient implies that if $x$ is $\ell$-fine, then a substring of $x$ can be matched advantageously with $y_\ell$.
\begin{lemma}
  Let $\ell$ be a power of two, and $w > \eps^{-2}\ell$. If a string $x$ of length $w$ is $\ell$-fine, then there exists an interval $I$ with $\LCS(x_I,y_\ell)\ge \frac{|I|}{2} + \eps^3|x|$. Furthermore, given $x$ and $\ell$, we can determine the interval $I$ and the common subsequence of $x_I$ and $y_\ell$ in time $O(w\log w)$.
  \label{lem:long-2}
\end{lemma}
\begin{proof}
  The number of $\ell^+$-flags in $x$ is at least $\eps w$. By pigeonhole, there exists some interval $I'=[a,b]$ of length $4\eps^2 w$ containing at least $2\eps^3w$ many $\ell^+$-flags (the lost factor of two accounts for $2\eps^2 w$ possibly not evenly dividing $w$).
  Furthermore, since $x$ is not coarse, we may assume that each such $\ell^+$-flag is an $\ell'$-flag for some $\ell'\in [\ell, \eps^2 w)$.
  Thus, the interval $I=[a, b+11\eps^2w]$ has length $4\eps^2 w + 11\eps^2w \le 20\eps^2 w$ and $x'\defeq x_{I}$ has at least $2\eps^3 w$ many $\ell^+$-flags (we cannot simply take $x'=x_{I'}$, as bits at the right end of $I$ may be flags in $x$ but not in $x_{I'}$). Furthermore, $x'$ is $\eps^2$-balanced since it has length at least $\eps^2 w$ and is a substring of $x$, which is not coarse.
  We can find interval $I'$, and thus $I$ and $x'$, in time $O(w\log w)$, because (with preprocessing of the string's prefix sums) we can test whether a bit is an $\ell^+$-flag in $\log w$ time, so counting the flags in an interval can be done in $O(w\log w)$ time, and there are a constant number of intervals to check.
 
  Now we claim $\LCS(x',y_\ell)\ge \frac{|x'|}{2} + \eps^3|x|$.
  Construct a common subsequence $x''$ of $x'$ and $y'$ as follows:
  Initialize a counter $i=1$. While $i\le 1(x')$, 
  \begin{enumerate}
  \item Append a one to $x''$,
  \item If the $i$th one of $x'$ is an $\ell'$-flag for some $\ell'\ge \ell$, append $1+\floor{10(\ell'-1)}$ zeros to $x''$ and $i\gets i+\ell'$.
  \item Otherwise $i\gets i+1$.
  \end{enumerate}
  We claim the subsequence $x''$ has the following properties.
  \begin{itemize}
  \item $x''$ is a subsequence of $x'$.
  \item $x''$ is a subsequence of $y_\ell$.
  \item $x''$ has length at least $\frac{|x'|}{2} + \eps^3|x|$.
  \end{itemize}

  To see the first property, take the subsequence of $x'$ where the one added to $x''$ when the counter is $i$ is matched to the $i$-th one of $x'$, and the zeros added when the counter is $i$ are the zeros between the $i$-th and $(i+\ell')$-th one of $x'$, of which there are at least $1+\floor{10(\ell'-1)}$ because $i$ is an $\ell'$-flag in $x'$.

  To see the second property, first note that $\floor{10(\ell'-1)} + 1\ge \ell'$ for all positive integers $\ell'$, so all runs of zeros in $x''$ have length at least $\ell$.
  Write $x''=1^{a_1}0^{a_2}1^{a_3}0^{a_4}\cdots 1^{a_{2k+1}}$, where all $a_{2i}\ge \ell$ and $a_{2i-1}\ge 1$ for all positive integers $i$ (except possibly $a_{2k+1}$, which may be 0).
  Notice that $1^{a_i}$ and $0^{a_i}$ are each subsequences of $(0^{\ell}1^{\ell})^{r_i}$, where $r_i\defeq \ceil{a_i/\ell}\le \frac{a_i}{\ell} + 1$.
  Thus, $x''$ is a subsequence of $(0^{\ell}1^{\ell})^{r}$ for $r\defeq r_1+\cdots+r_{2k+1}$.
  Thus, we have
  \begin{equation*}
    r \le r_1+\cdots+r_{2k+1}
    \le \frac{a_1+\cdots+a_{2k+1}}{\ell} + (2k+1)
    < \frac{4(a_1+\cdots+a_{2k+1})}{\ell}
    \le \frac{80\eps^2 w}{\ell} < \frac{\eps w}{5\ell},
  \end{equation*}
  proving that $x''$ is a subsequence of $y_\ell$. In the third inequality above, we used $a_2+a_4+\cdots+a_{2k}\ge k\ell$, so $2k+1\le 3k < \frac{3(a_1+\cdots+a_{2k+1})}{\ell}$.
  In the fourth inequality, we used $a_1+\cdots+a_{2k+1}=|x''| \le |x'| \le 20 \eps^2w$.

  To see the third property, notice that $|x''|-i$ only changes when a run of zeros is added to $x''$. If this run is added for an $\ell'$-flag at $i$ in $x'$, then difference $|x''|-i$ increases by $1+\floor{9(\ell'-1)} \ge \ell'$, while the total number of flags skipped over is at most $\ell'$.
  By induction on $i$, after every while-loop iteration, we have
  \begin{align}
  |x''|-i \ge \#\{\ell^+\text{-flags in $[i]$}\}.
  \end{align}
  so the total length of $x''$ at the end is at least 
  \[
  1(x')+\#\{\ell^+\text{-flags in $x'$}\} \ge \left(\frac{1}{2}-\eps^2\right)|x'|+2\eps^3w \ge \frac{|x'|}{2} + \eps^3|x|,
  \]
  as desired. The first inequality above follows from the fact that $x'$ is $\eps^2$-balanced, and the second from $|x'| \le 20\eps^2 w$.
\end{proof}

\subsection{Proof of the Algorithmic Structure Lemma}
\begin{proof}[Proof of Lemma~\ref{lem:struct}]
  Let $\delta_{code}=\eps^4/2$. We define the properties $P_t$ and $Q_t$ based on the types in Definition~\ref{def:type}.
  For convenience, we index them not as $P_1,P_2,\dots,P_T$, but rather as $P_{\ell,0}, P_{\ell, 1}, P_{\ell}$ for $\ell\le w$ a power of 2, for a total of roughly $T\sim 3\log w$ properties. 
\begin{itemize}
  \item If $\ell \ge \eps^2 w$ and $b\in\{0,1\}$, then $P_{\ell, b}$ is the property that $x$ is $\ell$-coarse, and its imbalanced interval of length $\ell$ is imbalanced in the direction of $b$-bits (i.e. has more $b$'s than $\bar{b}$'s).
  \item If $\ell \ge \eps^2 w$ and $b\in\{0,1\}$, then $Q_{\ell,b}$ is the property that $y$ has at least $\big(\frac{1+\eps^2}{2}\big)\ell$ $b$-bits.
  \item If $\ell < \eps^2 w$, then $P_\ell$ is the property that $x$ is $\ell$-fine.
  \item If $\ell < \eps^2 w$, then $Q_\ell$ is the property that $y$ contains $y_\ell = (0^{\ell}1^{\ell})^{\eps w/(5\ell)}$ as a subsequence.
  \end{itemize}

  We now verify the conditions of Lemma~\ref{lem:struct}.
  \begin{enumerate}
  \item \emph{For every length-$w$ string $x$, there exists a $t$ such that $x$ has property $P_t$.}
  
  This follows immediately from Lemma~\ref{lem:comb-struct}.

  \item \emph{If $\LCS(x,y)\ge (1-\delta_{code})|x|$ and $x$ has property $P_t$, then $y$ has property $Q_t$.}

  First suppose $\ell \ge \eps^2 w$, $b\in \{0,1\}$, and $x$ has property $P_{\ell,b}$. Then $x$ has a substring $x_I$ of length at least $\eps^2 w$ with at least $(\frac{1}{2}+\eps^2)\ell$ $b$-bits.
  The longest common subsequence of $x_I$ and $y$ has at least $|I|-\delta_{code}w$ of the bits of $x_I$, so $y$ has at least $(\frac{1}{2}+\eps^4)\ell - \delta_{code}w \ge (\frac{1+\eps^4}{2})\ell$ $b$-bits, satisfying property $Q_{\ell,b}$.

  If $x$ has property $P_{\ell}$ for $\ell < \eps^2 w$, $x$ is $\ell$-fine. By Lemma~\ref{lem:long-1}, $y$ has property $Q_{\ell}$.

  \item \emph{For every $t$, property $Q_t$ is hereditary, meaning that if $y$ has $Q_t$ and $y$ is a subsequence of $y'$, then $y'$ has $Q_t$.}

  This follows from the definition of $Q_t$ and that being a subsequence is a transitive relation.

  \item \emph{For every $t$, property $P_t$ can be tested in time $O(w\log w)$, and property $Q_t$ can be tested in time $O(w)$ on substrings of a string $y$, after $O(|y|\log |y|)$ prepreocessing.}

  Testing the $P_t$'s can be done in $O(w\log w)$ because, after $O(w)$ preprocessing by storing all prefix sums, we can check whether any particular index is an $\ell$-flag or not in $O(1)$ time, and for any particular property $P_t$, we need to check at most $O(w\log w)$ flags.

  To test $Q_t$, first note that if we are working with a coarse property $Q_{\ell,b}$, this can be tested in $O(1)$ time after preprocessing prefix sums. 
  To test a fine property $Q_\ell$, preprocess the string $y$ as follows: for every index $j \in\{ 0,1,\dots,|y|\}$ and bit $b\in\{0,1\}$, compute $\nxt_{b,\ell}(j)$, the smallest index $j'$ such that the substring $y_{[j,j']}$ has at least $\ell$ bits equal to $b$.
  For any $j$ and $b$, $\nxt_{b,\ell}(j)$ can be computed by a binary search in $\log(|y|)$ time, so the preprocessing takes time $O(|y|\log |y|)$.
  Now property $Q_t$ can be tested on a substring $y_J$ in $O(w)$ time by evaluating $\nxt_{1,\ell}(\nxt_{0,\ell}(\cdots \nxt_{1,\ell}(\nxt_{0,\ell}(\imin J))\cdots ))$, where there are $\varepsilon w / (5\ell)$ calls to each of $\nxt_{0,\ell}$ and $\nxt_{1,\ell}$, and checking if the result is at most $\imax J$.

  \item \emph{If $x$ has property $P_t$, $|x| = w$, and $y$ has property $Q_t$, then there exists an interval $I\subset[w]$ such that $\LCS(x_I,y)\ge \frac{|I|}{2}+\delta_{code} w$.
  Furthermore, given $x$ and $t$, the interval $I$ and the promised common subsequence of $x_I$ and $y$ can be chosen independent of $y$, and both can be found in time $O(w\log w)$.}

  First suppose $x$ has property $P_{\ell,b}$ and $y$ has property $Q_{\ell,b}$ for $\ell \ge \eps^2 w$ and $b \in \{0,1\}$.
  Then $x$ has a substring $x_I$ of length $\ell$ with at least $(\frac{1}{2}+\eps^2)\ell$ $b$-bits and $y$ has at least $(\frac{1+\eps^2}{2})\ell$ $b$-bits, so $\LCS(x_I,y)\ge (\frac{1+\eps^2}{2})\ell \ge \frac{|I|}{2} + \frac{\eps^4}{2} w$, as desired.
  Furthermore, $I$ can be found in linear time by a linear sweep, and the common subsequence is simply $b^{(\frac{1+\eps^2}{2})\ell}$ as desired.
  
  Now suppose $x$ has property $P_{\ell}$ for $\ell < \eps^2 w$, and $y$ has property $Q_{\ell}$. Thus, $x$ is $\ell$-fine and $y$ contains $y_\ell$ as a subsequence. By Lemma~\ref{lem:long-2}, there exists an interval $I$ with $\LCS(x_I,y)\ge \frac{|I|}{2} + \eps^3|x|$, as desired.
  Furthermore, also by Lemma~\ref{lem:long-2}, $I$ and the common subsequence of $x_I$ and $y$ can be computed from $x$ and $t$ in time $O(w\log w)$, independent of $y$, as desired.
  \end{enumerate}

  This proves Lemma~\ref{lem:struct}.
\end{proof}

%%%%%%%%%%%%%%%%%%%%%%%%%%%%%%%% LINEAR %%%%%%%%%%%%%%%%%%%%%%%%%%%%%%%%%%

\section{Almost-linear time algorithm} \label{sec:alg}

We now give the almost-linear time algorithm for the ``equally balanced'' case, which implies our main result. Specifically, we prove the following (see Section~\ref{app:imbalanced} for how Theorem~\ref{thm:main} follows from Theorem~\ref{thm:main-2-sketch}).

\begin{theorem*}[Theorem~\ref{thm:main-2-sketch}, restated]
  \thmmaintech
\end{theorem*}

The algorithm is given in Algorithm~\ref{alg:dp} with the covering step given in Algorithm~\ref{alg:cover}.
The rest of this section proves the correctness.

\subsection{Parameters and notation conventions}
Throughout $x$ and $y$ are the input strings satisfying $0(x) = 1(x) \le \min(0(y),1(y))$, and throughout $n = \max(|x|,|y|)$.
Let $w$ be the closest power of 2 to $\frac{|x|}{\log |x|}$.
We may assume by deleting a negligible number of bits from $x$ and $y$ that $|x|$ and $|y|$ are multiples of $w$.
Let $m_x\defeq \frac{|x|}{w} \sim \log |x|$ and $m_y \defeq \frac{|y|}{w} \sim \frac{|y|\log |x|}{|x|}$.
Throughout, we always denote intervals for string $x$ by the letter $I$, and intervals for string $y$ by the letter $J$.
By abuse of notation, we let intervals $I$ (possibly with decorations) denote the substring $x_I$, and we let intervals $J$ denote the substring $y_J$.

Let $\varepsilon>0$ be such that $O(n^{1+\varepsilon})$ is the desired runtime.
Let $\delta_{code}$ be the constant from Lemma~\ref{lem:struct}. 
Let $\delta_{eq}=\delta_{eq}(\frac{\varepsilon}{2})$ be the constant from Theorem~\ref{thm:rs20}.
Let $\alpha, \beta,\gamma,\delta,\theta$ be constant powers of $1/2$ that satisfy $\min(\delta_{eq},\delta_{code})\ge \alpha\gg \beta \gg \gamma\gg \delta=\theta$.
That is, $\delta$ is sufficiently small compared to $\gamma$, which is sufficiently small compared to $\beta$, which is sufficiently small compared to $\alpha$.
For completeness, we note it suffices to take $\theta=\delta =\gamma^8, \gamma=\frac{1}{2}\beta^2, \beta = \alpha^2$.
We did not try to optimize our constants.
We give the following intuition for the above parameters.
\begin{itemize}
\item $\alpha$ lower bounds the LCS advantage gained from both algorithmic structure lemma rectangles and nearly-square rectangles.
\item $\beta$ is the ``nearly-square'' parameter: in the optimal LCS, intervals $I$ are called {\it nearly-square} if they get matched to intervals of length $\le (1+\beta)|I|$.
We may assume at most $\beta^2$ fraction of intervals are nearly-square or else the nearly-square rectangles (together with ``trivial rectangles'') give a $(\frac{1}{2}+\poly\beta)$-approximation by applying $\EqLCS$ to each of them. 
\item $\gamma$ is the ``imbalanced'' parameter and discretization parameter for $I$-intervals: we may assume most $\gamma w$-length intervals to be $\gamma$-balanced, or else the ``Trivial rectangles'' give a $(\frac{1}{2}+\poly\gamma)$-approximation. We also round all $I$-intervals so that they are $\gamma w$-aligned. $\gamma$ is small enough so that the effect of this rounding is negligible.
\item $\delta$ is the overall LCS approximation advantage: we obtain a $(\frac{1 + \delta}{2})$-approximation. We assume $\LCS(x,y)\ge (1-\delta)|x|$, or else $\Trivial$ gives a $(\frac{1+\delta}{2})$-approximation.
\item $\theta$ is the discretization parameter for $J$-intervals: we round all $J$-intervals so that they are $\theta w$ aligned. $\theta$ is small enough so that the effect of this rounding is negligible. We take $\theta$, the $J$-interval discretization, to be smaller than $\gamma$, the $I$-interval discretization, so that the gain from matching ``Trivial rectangles'' is larger than the loss due to discretization.
\end{itemize}

\begin{algorithm}
  \caption{$\Covering$}
  \label{alg:cover}
  \KwIn{$x,y$ such that $1(x)=0(x) \le \min(1(y),0(y))$}
  \KwOut{A set $\cR$ of certified rectangles $(I\times J,\kappa)$ where $I$ is $\gamma w$-aligned, $J$ is $\theta w$-aligned, and $\LCS(I,J)\ge \kappa$.}
  
  \tcp*[l]{Trivial rectangles}
  $\cR \gets \{([0,|x|]\times [0,|y|], \Trivial([0,|x|],[0,|y|]))\}$\;
  \For{all $\gamma w$-aligned intervals $I$}{
    \For{$j=0,\dots,m_y/\theta$}{
      $J\gets $ the smallest $\theta w$-aligned interval s.t. $\imax J = \theta w j$ and $\Trivial(I,J)\ge (\frac{1}{2}-\sqrt{\delta}) |I|$\;
      \label{alg:easy-binarysearch-1}
      \If{$J$ exists}{
        $\cR\gets \cR\cup (I\times J, \Trivial(I,J))$\;
      }
      $J\gets $ the smallest $\theta w$-aligned interval s.t. $\imax J = \theta w j$ and $\Trivial(I,J)\ge (\frac{1}{2}+\frac{\gamma}{2}) |I|$\;
      \label{alg:easy-binarysearch-2}
      \If{$J$ exists}{
        $\cR\gets \cR\cup (I\times J, \Trivial(I,J))$\;
      }
    }
    \For{$\theta w$-aligned intervals $J$ with $|J|=|I|$} {
      $\cR\gets \cR\cup (I\times J,\Trivial(I,J))$\;
      \label{alg:trivial-square}
    }
  }

  \tcp*[l]{Nearly-square rectangles}
  \For{intervals $I\in\mathcal{I}_w$}{
    \For{$\theta w$-aligned intervals $J$ with $|J|\in[(1-\alpha)w,(1+\alpha)w]$}{
      $\cR\gets \cR\cup (I\times J, \EqLCS(I,J))$\;
      \label{alg:square}
    }
  }

  \tcp*[l]{Algorithmic structure lemma rectangles}
  \For{$i=1,\dots,m_x$}{
    $I\gets [(i-1)w,iw]$\;
    $t\gets \getP(x_I)$\;
    $I' \gets \getI(x_{I}, t)$\;
    \For{$j=1,\dots,m_y/\theta$}{
      $J\gets$ smallest interval such that $\imax J = \theta wj$,  $|J|\ge (1+0.9\beta)w$, and $\isQ(y_J,t)$\;
      \label{alg:binarysearch}
      \If{$J$ exists}{
        $\cR\gets \cR \cup (I'\times J, \frac{|I'|}{2} + \alpha w)$\;
      }
    }
  }

  \Return $\cR$.
\end{algorithm}

\begin{algorithm}
  \caption{$\FullLCS$}
  \label{alg:dp}
  \KwIn{$x,y$ such that $1(x)=0(x) \le \min(1(y),0(y))$}
  \KwOut{A $(\frac{1+\delta}{2})$-approximation of LCS}

  $\cR\gets \Covering(x,y)$\;
  \tcp*[l]{$\DP[i][j]$ lower bounds $LCS([0,\gamma wi],[0,\theta w j])$}
  $\DP \gets [0,m_x/\gamma]\times [0,m_y/\theta]$ array, initialized to 0\;
  \For{$i=1,\dots,m_x/\gamma$}{
    \For{$j=1,\dots,m_y/\theta$}{
      $\DP[i][j]\gets \max(\DP[i-1][j], \DP[i][j-1])$\;
      \For{$(I\times J,\kappa)\in\cR$ with $\imax I = (\gamma w)i, \imax J = (\theta w)j$,}{
        $\DP[i][j]\gets \max(\DP[i][j], \DP[(\imin I)/(\gamma w)][(\imin J)/(\theta w)] +\kappa)$\;
      }
    }
  }
  \Return $\DP[m_x/\gamma][m_y/\theta]$.
\end{algorithm}

\subsection{Runtime}
We now analyze the runtime. We re-emphasize, as we did in the proof sketch, that we need to be careful about factors $m_y$ in our runtime, but not powers of $m_x$: $m_x$ is only $\log n$, but $m_y$ is roughly $|y|/|x|$, which can be a positive power of $|y|$.

We first run a $O(|y|)$-time preprocessing of prefix-sums that allows us to query zero-counts and one-counts in any interval in either $x$ or $y$ in $O(1)$ time.
We also preprocess string $y$ so that we can test every property $Q_t$ efficiently on substrings of $y$; for each $t$ this takes $O(|y| \log |y|)$ preprocessing time, for a total preprocessing time that is $O(|y| \log ^2 |y|)$.

The runtime of $\Covering$ is dominated by calls to $\Trivial, \EqLCS, \getP,\isQ,$ and $\getI$.
Note that in Lines~\ref{alg:easy-binarysearch-1} and \ref{alg:easy-binarysearch-2}, $J$ can be computed by a binary search over a search space of size $m_y/\theta$, and thus can be found in $\log(m_y/\theta)$ calls to Trivial, which each take $O(1)$ time with preprocessing.
Thus, the first nested loop takes $O(m_x^2m_y\log m_y) \le \tilde O(|y|/|x|)$ time.

The second nested loop has $O(m_xm_y)$ calls to EqLCS, each of which runs in $O(|x|^{1+\frac{\varepsilon}{2}})$ time, and thus takes $O(n^{1+\varepsilon})$ time.

For the third nested loop, the number of calls to $\getP$ and $\getI$ is $m_x$, and each run in time $O(w\log w)$, so the runtime is at most $\tilde O(|x|)$.
Because the property $Q_t$ is hereditary, we can compute $J$ in Line~\ref{alg:binarysearch} by binary search with $\log(m_y/\theta)\le O(\log |y|)$ calls to $\isQ$, which runs in time $O(w)$ (the binary search crucially saves us a factor of roughly $|y|/|x|$ in the runtime).
The number of binary searches is $O(m_xm_y)$, so in total the runtime of this step is $O(m_xm_y\cdot \log(m_y)\cdot w) \le O(|y|\log^2 |y|)$.

There are $O(m_xm_y)$ rectangles, and the dynamic programming has $O(m_xm_y)$ states.
The runtime of the dynamic programming is thus $O(m_xm_y) \le \tilde O(|y|/|x|)$, so the total runtime is thus $O(|y|^{1+\varepsilon})$.

\subsection{Correctness proof, high level overview}

We need to show two things about our output, $\FullLCS(x,y)$.
\begin{align}
  \LCS(x,y)\ge \FullLCS(x,y)
  \label{eq:cor-easy}\\
  \frac{1+\delta}{2}\LCS(x,y)\le \FullLCS(x,y)
  \label{eq:cor-apx}
\end{align}
Equation \eqref{eq:cor-easy} is the easier, which we prove at the end of this section.
Equation \eqref{eq:cor-apx} is the harder direction.
We prove it in two cases, based on the following definition.
\begin{definition}
  We say a pair of binary strings $(x,y)$ is \emph{good}  if 
  \begin{enumerate}
    \item \label{i:lin-1} $\LCS(x,y)\ge (1-\delta)|x|$,
    \item \label{i:lin-2} For at least $(1-\gamma)m_x$ intervals $I\in \cI_w$, every $I'\in\cI_{\gamma w}(I)$ is $\gamma$-balanced, and
    \item \label{i:lin-3} At least $1-\beta^2$ of $I\in\cI_w$ satisfy $|J^\tau_I| \ge (1+\beta)|I|$.
  \end{enumerate}
  We call a pair \emph{bad} if it is not good.
  \label{def:nice}
\end{definition}
Obviously, we cannot determine in almost-linear time if an input is good or bad, since that involves computing $\LCS(x,y)$.
However, our analysis of the performance of $\FullLCS(x,y)$ differs depending on whether the input is good or bad.
In Section \ref{ssec:not-nice}, we prove \eqref{eq:cor-apx} when the input is bad, and in Section \ref{ssec:nice}, we prove \eqref{eq:cor-apx} when the input is good.

In both the easy direction \eqref{eq:cor-easy} and the hard direction \eqref{eq:cor-apx}, we use the following characterization of the output of the dynamic programming in Algorithm~\ref{alg:dp}. Recall a collection of rectangles is called an \emph{ordered} collection if every pair $(I, J)$ and $(I',J')$ is comparable (i.e. either $I<I'$ and $J<J'$ or $I > I'$ and $J>J'$).
\begin{lemma}
  The output of $\FullLCS$ is the maximum, over all ordered collections of certified rectangles $(I_1\times J_1,\kappa_1),\dots, (I_\ell\times J_\ell,\kappa_\ell)$, of $\kappa_1+\kappa_2+\cdots+\kappa_\ell$.
\label{lem:alg-0}
\end{lemma}
\begin{proof}
  By induction, it follows that $\DP[i][j]$ is the maximum, over all ordered collections of certified rectangles $(I_1\times J_1,\kappa_1),\dots, (I_\ell\times J_\ell,\kappa_\ell)$ contained in $[0,\gamma wi]\times [0,\theta w j]$, of $\kappa_1+\kappa_2+\cdots+\kappa_\ell$.
  Here, we use that, for all rectangles $I\times J$ in $\cR$, interval $I$ is $\gamma w$-aligned and interval $J$ is $\theta w$-aligned.
\end{proof}
The next lemma asserts that certified rectangles are indeed ``certified.''
\begin{lemma}
  Every certified rectangle $(I\times  J,\kappa)$ in $\Covering$ satisfies $\LCS(I,J)\ge \kappa$.
\label{lem:alg-1}
\end{lemma}
\begin{proof}
  This is true of all rectangles certified by Trivial and EqLCS by definition.
  The algorithmic structure lemma rectangles $(I'\times J,\kappa)$ for $\kappa=\frac{|I'|}{2}+\alpha w$ satisfy $\LCS(I',J)\ge \kappa$ by Lemma~\ref{lem:struct}.
\end{proof}
The easy direction \eqref{eq:cor-easy} follows easily from Lemma~\ref{lem:alg-0} and Lemma~\ref{lem:alg-1}.
\begin{corollary}
  $\FullLCS(x,y)\le \LCS(x,y)$
\end{corollary}
\begin{proof}
  By Lemma~\ref{lem:alg-0}, the output of $\FullLCS(x,y)$ equals $\kappa_1+\cdots+\kappa_\ell$ for some ordered collection of certified rectangles $(I_1\times J_1,\kappa_1),\dots, (I_\ell\times J_\ell,\kappa_\ell)$.
  Then, by Lemma~\ref{lem:alg-1}, we have 
  \begin{align}
    \FullLCS(x,y) 
    = \sum_{i=1}^{\ell} \kappa_i
    \le \sum_{i=1}^{\ell} \LCS(I_i,J_i)
    \le \LCS(x,y),
  \end{align}
  as desired.
\end{proof}

\subsection{Proof of (\ref{eq:cor-apx}) for bad inputs}
\label{ssec:not-nice}

We show that \eqref{eq:cor-apx} holds in the bad case by conditioning on which case of Definition~\ref{def:nice} is violated.

\paragraph{Subcase 1: Trivial.}
In the first subcase, we suppose $\LCS(x,y)\le (1-\delta)|x|$.
\begin{lemma}
  If $\LCS(x,y)\le (1-\delta)|x|$, then \eqref{eq:cor-apx} holds.
\label{lem:not-nice-1}
\end{lemma}
\begin{proof}
  We always have $\Trivial([0,|x|],[0,|y|]) \ge \frac{|x|}{2}$ as $\frac{|x|}{2}=1(x)=0(x)\le \min(1(y),0(y))$.
  Hence, we have $\FullLCS(x,y)\ge \frac{|x|}{2} \ge \frac{1+\delta}{2}\LCS(x,y)$.
\end{proof}

\paragraph{Subcase 2: Locally imbalanced.}
In the next subcase, we assume $\LCS(x,y)\ge (1-\delta)|x|$ and that a nontrivial fraction of intervals are imbalanced.
Since $x$ and $y$ have such a long LCS, we know that most intervals in $x$ appear nearly unmodified in $y$:
\begin{lemma}
If $w'$ is a positive integer that divides $|x|$ and $\LCS(x,y)\ge (1-\delta)|x|$, then at most $\sqrt{\delta}\frac{|x|}{w'}$ intervals $I_i \in \cI_{w'}$ satisfy $\LCS(I_i,J^\tau_{I_i})>(1-\sqrt{\delta})|I_i|$.
\label{lem:bad}
\end{lemma}
\begin{proof}
  To obtain the longest common subsequence of $x$ and $y$, one applies at most $\delta m$ deletions.
  By counting, at most $\sqrt{\delta}\frac{|x|}{w'}$ intervals of $\cI_{w'}$ receive more than $\sqrt{\delta}w'$ deletions, and the remaining intervals satisfy the desired inequality.
\end{proof}
We now can establish \eqref{eq:cor-apx} in this case.
\begin{lemma}
  If at least $\gamma m_x$ many $\gamma w$-intervals are $\gamma$-imbalanced, and $\LCS(x,y)\ge (1-\delta)|x|$, then \eqref{eq:cor-apx} holds.
\label{lem:not-nice-2}
\end{lemma}
\begin{proof}
  Let $I_1<\cdots<I_{m_x/\gamma}$ be the intervals of $\cI_{\gamma w}$.
  For all $i=1,\dots,m_x/\gamma$, let $J_i\defeq \round_{\theta w}(J^\tau_{I_i})$, so that $J_i$ are pairwise disjoint.
  Let $K_{imbal}$ be the indices $i$ such that $I_i$ is $\gamma$-imbalanced.
  By assumption $|K_{imbal}|\ge \gamma m_x$.
  Let $K_{good}$ be the indices $i$ such that $\LCS(I_i,J_{I_i}^\tau)\ge (1-\sqrt{\delta})|I_i|$.
  By Lemma~\ref{lem:bad} with $w' = \gamma w$, $|K_{good}|\ge (1-\sqrt{\delta}) \frac{m_x}{\gamma}$. 
  
  Observe that under these assumptions, $\Covering$ certifies many rectangles using the trivial algorithm. For $i\in K_{good}$, we have $\Trivial(I_i,J_i)\ge \frac{1}{2}(1-\sqrt{\delta})|I_i|-2\theta w \ge (\frac{1}{2}-\sqrt{\delta})|I_i|$.
  Thus, we certify $(I_i\times J_i',(\frac{1}{2}-\sqrt{\delta})|I_i|)$ for some subinterval $J_i'\subset J_i$, defined as the shortest $\theta w$-aligned interval with $\imax J_i'=\imax J_i$ and $\Trivial(I_i,J_i')\ge \frac{1}{2}(1-\sqrt{\delta})|I_i|$.

  For $i\in  K_{good}\cap K_{imbal}$, we have $\Trivial(I_i,J_i)\ge \Trivial(I_i,J_{I_i}^\tau) - 2\theta w \ge  (\frac{1}{2}+\gamma -\sqrt{\delta})|I_i|-2\theta w > (\frac{1}{2} + \frac{\gamma}{2})\gamma w$.
  Thus, we certify $(I_i\times J_i',(\frac{1}{2}+\frac{\gamma}{2})|I_i|)$ for some subinterval $J_i'\subset J_i$, defined as the shortest $\theta w$-aligned interval with $\imax J_i'=\imax J_i$ and $\Trivial(I_i,J_i')\ge (\frac{1}{2}+\frac{\gamma}{2})|I_i|$.
  
  Thus $\Covering$ obtains an ordered collection of certified rectangles $(I_i\times J_i',\kappa_i)$ over $i\in K_{good}$ with $\kappa_i\ge (\frac{1}{2}-\sqrt{\delta})|I_i|$ for all $i\in K_{good}$ and $\kappa_i\ge (\frac{1}{2}-\sqrt{\delta}+\frac{\gamma}{2})|I_i|$ for all $i\in K_{good}\cap K_{imbal}$. 
  Thus,
  \begin{align}
    \FullLCS(x,y)
    &\ge\sum_{i\in K_{good}} \kappa_i \nonumber\\
    &\ge 
    \left(\frac{1}{2}-\sqrt{\delta}\right)(\gamma w)\cdot |K_{good}|
    + \frac{\gamma}{2} (\gamma w)\cdot |K_{good}\cap K_{imbal}|\nonumber\\
    &\ge \left( \frac{1+\delta}{2}\right)\LCS(x,y),
  \end{align}
  as desired.
  In the last inequality, we used (i) $|K_{good}|\ge (1-\sqrt{\delta})\frac{m_x}{\gamma}$, (ii) $|K_{good}\cap K_{imbal}|\ge \gamma m_x - \frac{\sqrt{\delta}m_x}{\gamma}$, (iii) $\gamma\gg \delta$, and (iv) $m_x w = |x| \ge \LCS(x,y)$.
\end{proof}

\paragraph{Subcase 3: Many nearly-square intervals.}

This case applies when the equal-length input algorithm \cite{RubinsteinS20}  correctly certifies many rectangles.
Recall that an interval $I\in\cI_w$ is \emph{nearly-square} if $|J^\tau_I|\le (1+\beta)|I|$. For convenience, we call $I$ {\it oblong} if it is not nearly-square.
\begin{lemma}
  If at least $\beta^2m_x$ intervals $I\in\cI_w$ are nearly-square, then \eqref{eq:cor-apx} holds.
\label{lem:not-nice-3}
\end{lemma}
\begin{proof}
  Let $I_1<I_2<\dots<I_{m_x}$ be the intervals of $\mathcal{I}_w$.
  For all $i=1,\dots,m_x$, let $J_i\defeq\round_{\theta w}(J^\tau_{I_i})$, so that the $J_i$ are pairwise disjoint.
  Let $K_{short}$ be the set of indices $i$ such that $I_i$ is nearly-square.
  By assumption, $|K_{short}|\ge \beta^2m_x$.
  Let $K_{good}$ be the set of indices $i$ such that $\LCS(I_i,J_{I_i}^\tau)\ge (1-\sqrt{\delta})|I_i|$.
  By Lemma~\ref{lem:bad}, $|K_{good}|\ge (1-\sqrt{\delta}) m_x$.

  Just as in the proof of Lemma~\ref{lem:not-nice-3}, we track the rectangles that are certified by $\Covering$. For $i\in K_{good}$, we have $\Trivial(I_i,J_i)\ge (\frac{1}{2}-\sqrt{\delta})|I_i|$, so we certify $(I_i\times J_i',(\frac{1}{2}-\sqrt{\delta})|I_i|)$ for some subinterval $J_i'\subset J_i$.
  For $i\in K_{good}\cap K_{short}$, we have $(1+\alpha)w\ge (1+\beta)w \ge |J^\tau_{I_i}| \ge |J_i|$ since $I_i$ is nearly-square, and $|J_i|\ge \LCS(I_i,J_i)\ge\LCS(I_i,J_{I_i}^\tau)-2\theta w\ge (1-\sqrt{\delta} - 2\theta)w \ge (1-\alpha)w$.
  Hence, $\EqLCS(I_i,J_i)$ is called at Line~\ref{alg:square} and has value at least $(\frac{1}{2}+\alpha)\LCS(I_i,J_i) \ge (\frac{1}{2}+\alpha)(1 - \sqrt{\delta}-2\theta)w > \frac{1+\alpha}{2} w$ by Theorem~\ref{thm:rs20}.

  We thus have an ordered collection of certified rectangles $(I_i\times J_i',\kappa_i)$ over $i\in K_{good}$ with $\kappa_i\ge (\frac{1}{2}-\sqrt{\delta})w$ for all $i\in K_{good}$ and $\kappa_i\ge (\frac{1}{2}-\sqrt{\delta}+\frac{\alpha}{2})w$ for $i\in K_{good}\cap K_{short}$. 
  Thus,
  \begin{align}
    \FullLCS(x,y)
    &\ge\sum_{i\in K_{good}} \kappa_i \nonumber\\
    &\ge 
    \left(\frac{1}{2}-\sqrt{\delta}\right)w\cdot |K_{good}|
    + \frac{\alpha}{2}  w\cdot |K_{good}\cap K_{short}|\nonumber\\
    &\ge \left( \frac{1+\delta}{2}\right)\LCS(x,y),
  \end{align}
  as desired.
  In the last inequality, we used (i) $|K_{good}|\ge (1-\sqrt{\delta})m_x$, (ii) $|K_{good}\cap K_{short}|\ge \beta^2 m_x - \sqrt{\delta}m_x$, (iii) $\alpha\gg\beta\gg \delta$, and (iv) $m_x w = |x| \ge \LCS(x,y)$.
\end{proof}

\paragraph{Wrapping up the bad case.}
We now can prove the following lemma.
\begin{lemma}
  If $(x,y)$ is bad, then \eqref{eq:cor-apx} holds.
\label{lem:not-nice}
\end{lemma}
\begin{proof}
  If $(x,y)$ is bad, then either item 1, 2, or 3 of Definition~\ref{def:nice} is violated. 
  If 1 is violated, \eqref{eq:cor-apx} holds by Lemma~\ref{lem:not-nice-1}. 
  If 2 is violated, there are at least $\gamma m_x$ intervals $I$ with a $\gamma$-imbalanced $\gamma w$-subinterval, so there are at least $\gamma m_x$ many $\gamma$-imbalanced  $\gamma w$-intervals, so by Lemma~\ref{lem:not-nice-2}, \eqref{eq:cor-apx} holds.
  If 3 is violated, \eqref{eq:cor-apx} holds by Lemma~\ref{lem:not-nice-3}. 
\end{proof}

\subsection{Proof of (\ref{eq:cor-apx}) for good inputs}
\label{ssec:nice}
Let
\begin{align}
  m_x' \defeq (1-2\beta^2)m_x.
\end{align}
The following lemma establishes the natural structural property for good inputs.
\begin{lemma}
  If $(x,y)$ is good, then there exist an ordered sequence of rectangles $I_1\times J_1<\cdots<I_{m_x'}\times J_{m_x'}$ such that for all $i$, (i) $I_i\in\mathcal{I}_w$, (ii) every $\gamma w$-subinterval of $I_i$ is $\gamma$-balanced, (iii) $|J_i|\ge (1+0.9\beta)w$, and (iv) $\isQ(y_{J_i},\getP(x_{I_i}))$ returns true.
\label{lem:nice-0}
\end{lemma}
\begin{proof} 
  Among the $m_x$ intervals $I\in \mathcal{I}_w$, all but at most $\sqrt{\delta}m_x$ intervals satisfy $\LCS(I,J^\tau_{I})>(1-\sqrt{\delta})|I|$ by Lemma~\ref{lem:bad}, at most $\gamma m_x$ have a $\gamma$-imbalanced $\gamma w$-subinterval since $(x,y)$ is good, and at most $\beta^2 m_x$ are nearly-square since $(x,y)$ is good.
  Thus, at least $(1-2\beta^2)m_x = m'_x$ intervals are (i) satisfying $\LCS(I,J^\tau_{I})>(1-\sqrt{\delta})|I|$, (ii) $\gamma$-balanced in all $\gamma w$-subintervals, and (iii) oblong.
  Let $I_1<\cdots<I_{m_x'}$ be $m_x'$ of these intervals.
  Let $J_i \defeq \round_{\theta w}(J^\tau_{I_i})$, so these $J_i$ are pairwise disjoint.
  For all such $i$, we have
  \begin{align}
    |J_i|\ge |J^\tau_{I_i}|-2\theta w \ge (1+\beta-2\theta)w \ge (1+0.9\beta)w
    \label{eq:lb-2}
  \end{align}
  and
  \begin{align}
  \LCS(I_i,J_i) 
  \ge \LCS(I_i,J^\tau_{I_i}) - 2\theta w 
  \ge (1-\sqrt{\delta} - 2\theta) w
  \ge (1-\alpha)w.
  \end{align}
  For all $t$ such that $x_{I_i}$ has property $P_t$, we have $y_{J_i}$ has property $Q_t$ by Lemma~\ref{lem:struct} (Item 2).
  Thus, $\isQ$($y_{J_i},t$) returns true for $t=\getP(x_{I_i})$.
  We have found our ordered sequence of rectangles $I_1\times J_1<\cdots<I_{m_x'}\times J_{m_x'}$.
\end{proof}

We now can prove the main result for this section.
\begin{lemma}
  If $(x,y)$ is good, then \eqref{eq:cor-apx} holds.
\label{lem:nice-1}
\end{lemma}
\begin{proof}
  Let $I_1\times J_1< I_2\times J_2<\cdots<I_{m_x'}\times J_{m_x'}$ be the ordered sequence of rectangles given by Lemma~\ref{lem:nice-0}.
  By construction, for all $i=1,\dots,m_x'$, we have (i) $I_i\in\mathcal{I}_w$, (ii) every $\gamma w$-subinterval of $I_i$ is $\gamma$-balanced, (iii) $|J_i|\ge (1+0.9\beta)w$, and (iv) $\isQ(y_{J_i}, \getP(x_{I_i}))=\mathsf{true}$.

  It follows that at loop iteration $i=(\imax I_i)/w$ and $j=(\imax J_i)/(\theta w)$ of Line~\ref{alg:binarysearch}, the interval $J$ exists (the interval $J=J_i$ satisfies the requirement, so a minimal $J$ exists).
  Thus, $\Covering$ certifies a rectangle $(I_i'\times J_i', \frac{|I_i'|}{2}+\alpha w)$ where $I_i'$, the output of $\getI$, is a $\gamma w$-aligned subinterval of $I_i$ and $J_i'$ is a subinterval of $J_i$ with length at least $(1+0.9\beta)w$.

  We would like to build an ordered collection of certified rectangles containing these $(I_i'\times J_i', \frac{|I_i'|}{2}+\alpha w)$, which embed more than half of each small interval $I_i'$ into $y$. However, for each of these rectangles, $J_i'$ is typically much longer than $I_i'$, so using many of them is extremely wasteful of bits in $y$. To reduce this issue, we let $t\defeq 3/\beta$ and build an ordered collection using only every $t$-th rectangle from the preceding family.
  
  Let $m_x''$ be the largest multiple of $t$ less than $m_x'$.
  For each $i$ that is a multiple of $t$, partition $I_i$ into  $\tilde I_i^L< \tilde I_i^M< \tilde I_i^R$ where $\tilde I_i^M\defeq I_i'$.
  For $i$ not a multiple of $t$, let $\tilde I_i \defeq I_i$.
  For $i$ a multiple of $t$, let $\tilde J_i^M \defeq J_i'$.

  For $i$ a multiple of $t$, we claim there exist $\theta w$-aligned intervals $\tilde J_{i-t}^R< \tilde J_{i-t+1}<\tilde J_{i-t+2}<\cdots<\tilde J_{i-1}<\tilde J_{i}^L$ such that
  \begin{itemize}
  \item $\tilde J^L_{i} < \tilde J^M_i$.
  \item $|\tilde J_{i}^L|=|\tilde I_{i}^L|$.
  \item $|\tilde J_{i-\ell}|=|\tilde I_{i-\ell}|$ for $\ell=1,\dots,t-1$.
  \item $|\tilde J_{i-t}^R|=|\tilde I_{i-t}^R|$ (we take $\tilde I_{0}^R  = \emptyset$).
  \item $\tilde J^M_{i-t} < \tilde J^R_{i-t}$ (this is vacuously true if $i=t$)
  \end{itemize}
  To see that such intervals exist, notice that interval $[\imax\tilde J_{i-t}^M,\imin \tilde J_i^M] = [\imax J_{i-t}', \imin J_i']$ contains all intervals $J_{i-\ell}'$ for $\ell=1,\dots,t-1$. Since each $J_{i-\ell}'$ has length at least $(1+0.9\beta)w$, we have
  \begin{align}
    \imin J_i^M-\imax J_{i-t}^M
    \ge \left( t-1 \right)(1+0.9\beta)w
    > \left(t+1\right)w
    \ge |\tilde I_i^L| +\sum_{\ell=1}^{t-1} |\tilde I_{i-\ell}| + |\tilde I_{i-t}^R|.
  \label{}
  \end{align}
  The last inequality holds as each term on the right is at most $w$.
  Thus we can construct the intervals greedily in order $\tilde J_i^L, \tilde J_{i-1},\tilde J_{i-2},\dots,\tilde J_{i-t+1},\tilde J_{i-t}^R$ by setting $\imax \tilde J_i^L=  \imin \tilde J_i^M$, and then $\imax \tilde J_{i-1} = \imin \tilde J_i^L$, and so on.
  They will be $\theta w$-aligned as all of the $\tilde I$ intervals have lengths a multiple of $\gamma w$, and thus a multiple of $\theta w$.

  By construction of these intervals, $\Covering$ certifies the following rectangles for $i\le m_x''$:
  \begin{align}
    &(\tilde I_i^M\times \tilde J_i^M,\kappa_i^M) &&\text{where }\kappa_i^M \defeq \frac{|\tilde I_i^M|}{2} + \alpha w && \text{if $t \mid i$} \nonumber\\
    &(\tilde I_i^L\times \tilde J_i^L,\kappa_i^L) && \text{where }\kappa_i^L \defeq \Trivial(\tilde I_i^L,\tilde J_i^L) && \text{if  $t \mid i$} \nonumber\\
    &(\tilde I_i^R\times \tilde J_i^R,\kappa_i^R) && \text{where }\kappa_i^R \defeq \Trivial(\tilde I_i^R,\tilde J_i^R)&& \text{if  $t \mid i$} \nonumber\\
    &(\tilde I_i\times \tilde J_i,\kappa_i) &&\text{where }\kappa_i \defeq \Trivial(\tilde I_i,\tilde J_i) && \text{if $t \not \mid i$} 
  \label{eq:nice-rectangles}
  \end{align}
  The first collection of rectangles comes from the definition of $\tilde I_i^M\defeq I_i'$ and $\tilde J_i^M\defeq J_i'$.
  The rest of the rectangles come from the fact that we certify all $\gamma w$-aligned squares with the trivial algorithm in $\Covering$ Line~\ref{alg:trivial-square}.
  Furthermore, the rectangles are increasing in $i$, with additionally $\tilde I_i^L\times \tilde J_i^L < \tilde I_i^M\times \tilde J_i^M < \tilde I_i^R\times \tilde J_i^R$ for $i$ a multiple of $t$.
  Hence, the rectangles in \eqref{eq:nice-rectangles} form an ordered collection of rectangles.
  By Lemma~\ref{lem:balance}, we also have $\kappa_i^L \ge (\frac{1}{2}-\gamma)|\tilde I_i^L|, \kappa_i^R \ge (\frac{1}{2}-\gamma)|\tilde I_i^R|$ for $i$ a multiple of $t$ and $\kappa_i \ge (\frac{1}{2}-\gamma)|\tilde I_i|$ for all other $i$, because the intervals $\tilde I_i^R,\tilde I_i^L,\tilde I_i$ are all $\gamma w$-aligned and thus $\gamma$-balanced.
  Thus, by Lemma~\ref{lem:alg-0}, the output of $\FullLCS(x,y)$ is at least
  \begin{align}
    \sum_{\substack{i\le m_x'' \\ t| i}}^{} \left(\kappa_i^L + \kappa_i^M+ \kappa_i^R\right) +
    \sum_{\substack{i\le m_x'' \\ t\nmid i}}^{} \kappa_i 
    &\ge 
    \sum_{\substack{i\le m_x'' \\ t| i}}^{} \left(\left( \frac{1}{2}-\gamma \right)(|\tilde I_i^L| + |\tilde I_i^M|+ |\tilde I_i^R| )  + \alpha w\right) +
    \sum_{\substack{i\le m_x'' \\ t\nmid i}}^{}  \left( \frac{1}{2}-\gamma \right)|\tilde I_i| \nonumber\\
    &\ge \left( \frac{1}{2}-\gamma \right)w\cdot m_x'' + \alpha w\cdot \frac{\beta}{3}m_x'' \nonumber\\
    &\ge \left( \frac{1+\delta}{2} \right)\LCS(x,y)
  \label{}
  \end{align}
  In the third inequality, we used that $m_x''w \ge (m_x'- t)w \ge (1-3\beta^2)m_xw$ and $m_xw = |x|\ge \LCS(x,y)$.
\end{proof}

We can now finish the proof of Theorem \ref{thm:main-2-sketch}.
\begin{proof}[Proof of Theorem \ref{thm:main-2-sketch}]
We have now proved that \eqref{eq:cor-easy} and \eqref{eq:cor-apx} always hold, and that $\FullLCS$ runs in time $O(n^{1+\varepsilon})$, so $\FullLCS$ gives a $(\frac{1+\delta}{2})$-approximation of the LCS of two binary strings with $0(x)=1(x)\le \min(0(y),1(y))$ in time $O(n^{1+\varepsilon})$, as desired.
\end{proof}

\section{Putting it all together}
\label{app:imbalanced}

In this final section we use standard techniques to finish the proof of Theorem~\ref{thm:main} given the balanced case Theorem~\ref{thm:main-2-sketch}. This proved in \cite{RubinsteinS20} for equal length strings and in \cite{AkmalV21} for unequal length strings (see also \cite{Akmal21}).

\begin{lemma}[Lemma 13 of \cite{AkmalV21}, see also Lemma 3.5 of \cite{Akmal21}]
For every $\rho >0$, there exists $\delta=\delta(\rho)>0$ such that the following holds. There exists an algorithm which, given binary strings $x,y$ with $|x|\le |y|$ and $0(x)=1(y)\le(\frac{1}{2} -\rho)|x|$, computes a $(\frac{1}{2} + \delta)$-approximation of $\LCS(x,y)$ in deterministic linear time.
\footnote{There are several minor differences between this statement and the statement in \cite{AkmalV21}. 

First, the statement in \cite{AkmalV21} says subquadratic time but it actually runs in linear time, similar to the analogous algorithm in \cite{RubinsteinS20} who proved Lemma~\ref{lem:av20} for equal-length strings. This was confirmed in private communication with the authors.

Second, \cite{AkmalV21} prove the statement when $0(x)$ and $1(y)$ are within $\varepsilon |x|$ of each other for some $|x|$, while we only consider when they are equal.
}
\label{lem:av20}
\end{lemma}

\begin{lemma}
  For all $\varepsilon>0$, there exists an absolute constant $\delta=\delta(\varepsilon)>0$ and a deterministic algorithm that, given two strings $x$ and $y$ with $|x|\le |y|$ and $\min(1(x),1(y)) = \min(0(x),0(y))$, outputs a $(\frac{1}{2} + \delta)$-approximation of $\LCS(x,y)$ in time $O(|y|^{1+\varepsilon})$.
\label{lem:main-app}
\end{lemma}
\begin{proof}
  Let $\delta_1=\delta_1(\varepsilon)>0$ be the absolute constant in Theorem~\ref{thm:main-2-sketch}.
  Let $\delta_2>0$ be the absolute constant in Lemma~\ref{lem:av20} with parameter $\rho=\delta_1/10$.
  Let $\delta = \min(\delta_1/2,\delta_2)$.
  As $|x|\le |y|$, we have three cases, and we find a $(\frac{1}{2}+\delta)$-approximation to $\LCS(x,y)$ in each.

  \paragraph{Case 1. $0(x)=\min(0(x),0(y)), 1(x)=\min(1(x),1(y))$.}
  Then $0(x)=1(x) = |x|/2$ and $\LCS(x,y) \ge |x|/2$.
  The algorithm in Theorem~\ref{thm:main-2-sketch}, gives a $(\frac{1}{2}+\delta_1)$-approximation of the LCS.

  \paragraph{Case 2. $0(y)=\min(0(x),0(y)), 1(x)=\min(1(x),1(y))$.}
  We have $1(x)=0(y)\le 0(x)$. There are two subcases.
  
  \textbf{Subcase 2a. $1(x)\ge (\frac{1}{2}-\rho)|x|$.}
  In this case, delete $0(x)-1(x)\le \rho |x|$ zeros from $x$ arbitrarily to get a balanced subsequence $x'$.
  Then $\LCS(x',y) \ge \LCS(x,y) - \rho |x|\ge (1-\rho)\LCS(x,y)$.
  Thus, the algorithm in Theorem~\ref{thm:main-2-sketch} gives a common subsequence of length $(\frac{1}{2} + \delta_1)(1-\rho)\LCS(x,y) \ge (\frac{1}{2} + \delta)\LCS(x,y)$.

  \textbf{Subcase 2b. $1(x)\le (\frac{1}{2}-\rho)|x|$.}
  In this case, Lemma~\ref{lem:av20} with parameter $\rho$ finds a common subsequence of length at least $(\frac{1}{2}+\delta)\LCS(x,y)$.

  \paragraph{Case 3. $0(x)=\min(0(x),0(y)), 1(y)=\min(1(x),1(y))$.}
  Symmetric to case 2.
\end{proof}

\begin{theorem*}[Theorem~\ref{thm:main}, restated]
  \thmmain
\label{thm:main-restate}
\end{theorem*}
\begin{proof}
  Let $\delta_0$ be the constant in Lemma~\ref{lem:main-app}. Let $\delta=\delta_0/5$.
  Let the input strings be $x$ and $y$ and assume without loss of generality $|x|\le |y|$ and that $\min(0(x),0(y)) \ge \min(1(x),1(y))$.
  We have $\Trivial(x,y) = \min(0(x),0(y))$ and $\LCS(x,y)\le \min(0(x),0(y)) + \min(1(x),1(y))$.

  If $\min(0(x),0(y)) \ge (1+\delta_0)\min(1(x),1(y))$, then because $\frac{1+\delta_0}{2+\delta_0} > \frac{1}{2} + \delta$ the trivial algorithm gives a $(\frac{1}{2} + \delta)$-approximation of the LCS. Thus we may assume $\min(0(x),0(y)) \le (1+\delta_0)\min(1(x),1(y))$.
  Delete $\min(0(x),0(y))-\min(1(x),1(y))\le \delta_0\min(1(x),1(y))$ zeros from each of $x$ and $y$ arbitrarily to obtain $x'$ and $y'$ with $\min(0(x'),0(y'))=\min(1(x'),1(y'))$.
  We have
  \begin{align}
  \LCS(x',y')\ge \LCS(x,y)-\delta_0\min(1(x), 1(y))\ge (1-\delta_0)\LCS(x,y).
  \end{align}
  Running the algorithm in Lemma~\ref{lem:main-app} gives an approximation to $\LCS(x',y')$ that is at least $(\frac{1}{2} + \delta_0)(1-\delta_0)\LCS(x,y) > (\frac{1}{2} + \delta)\LCS(x,y)$, as desired.
\end{proof}

\section{Conclusion and open questions}
We close with some related open questions.
\begin{itemize}
\item What is the best possible approximation factor of binary LCS in almost-linear or truly subquadratic time?
We give a $\frac{1}{2}+\delta$ in almost-linear time. We made no attempt to optimize $\delta$, and currently it depends on the runtime exponent $1+\varepsilon$.

\item Related to the above, can we prove fine-grained hardness of approximation results for LCS?
It is known that a deterministic $2^{-(\log n)^{1-\delta}}$ approximation in $n^{2-\eps}$ time for LCS over alphabet $n^{o(1)}$ would imply new circuit lower bounds, as would a deterministic $1-\frac{1}{\poly\log n}$-approximation for binary inputs \cite{AbboudB17, AbboudR18, ChenGLRR19}.

\item We studied the algorithmic question of computing LCS, where, as the previous two questions highlight, the optimal approximation factor is open. We showed this algorithmic question is closed related to an analogous combinatorial question, which is also open: What is the largest constant $\alpha\in(0,1)$ such that in any set $C\subset\{0,1\}^n$ of $|C|\ge 2^{\Omega(n)}$ binary strings, there are always two strings $x,y$ with $\LCS(x,y)\ge \alpha n$? The optimal $\alpha$ is known to be in $[\frac{1}{2}+10^{-40},2-\sqrt{2}]$ \cite{GuruswamiHL22, BukhGH17}, and $1-\alpha$ quantifies the maximum fraction of adversarial deletions that can be tolerated by a (asymptotically) positive rate code.

It would also be interesting to understand how strong is the connection between the deletion codes question and the algorithmic LCS question.
At first blush, it seems that techniques derived solely from analysis of deletion codes should not give an $\alpha$-approximation for $\alpha>2-\sqrt{2}\approx 0.59$ (because of the deletion codes construction \cite{BukhGH17}), so beating this ratio would show some separation between the two questions.

\item How does the optimal subquadratic time or almost-linear time approximation factor grow with the alphabet size? Over alphabet size $q$, we show that we can beat (barely) the trivial $\frac{1}{q}$-approximation. We know that we can always get a randomized $\frac{1}{n^{o(1)}}$-approximation in linear time \cite{AndoniNSS22,Nosatzki21}, which is much better than $\frac{1}{q}$ for large alphabets.
\item There is a natural question that arises from another possible approach to proving Theorem~\ref{thm:main}.
Define the \emph{directed edit distance} of two strings $x,y$ to be the number of edits needed to get from $x$ to $y$, where insertions cost 0 and deletions (and substitutions) cost 1.
Equivalently, $\dedit(x,y)\defeq |x|-\LCS(x,y)$. 
When the strings are equal length, the directed edit distance is simply half the edit distance.
A constant-factor approximation of directed edit distance in almost-linear time would immediately imply  Theorem~\ref{thm:main-2-sketch} and thus Theorem~\ref{thm:main}.
This suggest the following question, which may be of independent interest.
\begin{question}
  Is there an almost-linear time constant-factor approximation of the directed edit distance?
\end{question}
We note that $\dedit(x,y)$ is \emph{not} a metric. Indeed, it is not even symmetric\footnote{$\dedit(0011,00)=2$ but $\dedit(00,0011)=0$}, and it does not satisfy the triangle inequality.
Thus, the existing edit distance approximation algorithms \cite{CDGKS20, BrakensiekR20, KouckyS20, AndoniN20}, which rely heavily on the triangle inequality, do not seem to immediately apply to directed edit distance.
On the other hand, directed edit distance does satisfy a ``directed triangle inequality'': for strings $x,y,z$, we have $\dedit(x,z) \le \dedit(x,y) + \dedit(y,z)$. 
This gives some hope that fast approximation algorithms exist.
\end{itemize}

\section{Acknowledgements}

We thank Saeed Seddighin for introducing us to the question of approximating binary LCS and for suggesting its potential connection to the deletion codes bound \cite{GuruswamiHL22}.
We thank Negev Shekel Nosatzki for helpful discussions about edit-distance algorithms.
We thank Shyan Akmal and Virginia Vassilevska-Williams for helpful discussions on their work \cite{AkmalV21}.
We thank Aviad Rubinstein for helpful feedback.
We thank Venkatesan Guruswami for helpful feedback and encouragement.

\bibliographystyle{alpha}
\bibliography{bib}

\newcommand{\etalchar}[1]{$^{#1}$}
\begin{thebibliography}{AHWW16}

\bibitem[AB17]{AbboudB17}
Amir Abboud and Arturs Backurs.
\newblock Towards hardness of approximation for polynomial time problems.
\newblock In Christos~H. Papadimitriou, editor, {\em 8th Innovations in
  Theoretical Computer Science Conference, {ITCS} 2017, January 9-11, 2017,
  Berkeley, CA, {USA}}, volume~67 of {\em LIPIcs}, pages 11:1--11:26. Schloss
  Dagstuhl - Leibniz-Zentrum f{\"{u}}r Informatik, 2017.

\bibitem[ABW15]{AbboudBW15}
Amir Abboud, Arturs Backurs, and Virginia~Vassilevska Williams.
\newblock Quadratic-time hardness of {LCS} and other sequence similarity
  measures.
\newblock {\em CoRR}, abs/1501.07053, 2015.

\bibitem[AHWW16]{AbboudHWW16}
Amir Abboud, Thomas~Dueholm Hansen, Virginia~Vassilevska Williams, and Ryan
  Williams.
\newblock Simulating branching programs with edit distance and friends: or: a
  polylog shaved is a lower bound made.
\newblock In Daniel Wichs and Yishay Mansour, editors, {\em Proceedings of the
  48th Annual {ACM} {SIGACT} Symposium on Theory of Computing, {STOC} 2016,
  Cambridge, MA, USA, June 18-21, 2016}, pages 375--388. {ACM}, 2016.

\bibitem[Akm21]{Akmal21}
Shyan Akmal.
\newblock Longest common subsequence over constant-sized alphabets: Beating the
  naive approximation ratio.
\newblock Master's thesis, Massachusetts Institute of Technology, 2021.

\bibitem[AKO10]{AndoniKO10}
Alexandr Andoni, Robert Krauthgamer, and Krzysztof Onak.
\newblock Polylogarithmic approximation for edit distance and the asymmetric
  query complexity.
\newblock In {\em 51th Annual {IEEE} Symposium on Foundations of Computer
  Science, {FOCS} 2010, October 23-26, 2010, Las Vegas, Nevada, {USA}}, pages
  377--386. {IEEE} Computer Society, 2010.

\bibitem[AN20]{AndoniN20}
Alexandr Andoni and Negev~Shekel Nosatzki.
\newblock Edit distance in near-linear time: it's a constant factor.
\newblock In Sandy Irani, editor, {\em 61st {IEEE} Annual Symposium on
  Foundations of Computer Science, {FOCS} 2020, Durham, NC, USA, November
  16-19, 2020}, pages 990--1001. {IEEE}, 2020.

\bibitem[And18]{Andoni18}
Alexandr Andoni.
\newblock Simpler constant-factor approximation to edit distance problems.
\newblock 2018.

\bibitem[ANSS22]{AndoniNSS22}
Alexandr Andoni, Negev~Shekel Nosatzki, Sandip Sinha, and Clifford Stein.
\newblock Estimating the longest increasing subsequence in nearly optimal time.
\newblock 2022.

\bibitem[AO12]{AndoniO12}
Alexandr Andoni and Krzysztof Onak.
\newblock Approximating edit distance in near-linear time.
\newblock {\em {SIAM} J. Comput.}, 41(6):1635--1648, 2012.

\bibitem[AR18]{AbboudR18}
Amir Abboud and Aviad Rubinstein.
\newblock Fast and deterministic constant factor approximation algorithms for
  lcs imply new circuit lower bounds.
\newblock In {\em 9th Innovations in Theoretical Computer Science Conference
  (ITCS 2018)}. Schloss Dagstuhl-Leibniz-Zentrum fuer Informatik, 2018.

\bibitem[AW21]{AkmalV21}
Shyan Akmal and Virginia~Vassilevska Williams.
\newblock Improved approximation for longest common subsequence over small
  alphabets.
\newblock In Nikhil Bansal, Emanuela Merelli, and James Worrell, editors, {\em
  48th International Colloquium on Automata, Languages, and Programming,
  {ICALP} 2021, July 12-16, 2021, Glasgow, Scotland (Virtual Conference)},
  volume 198 of {\em LIPIcs}, pages 13:1--13:18. Schloss Dagstuhl -
  Leibniz-Zentrum f{\"{u}}r Informatik, 2021.

\bibitem[AWW14]{AbboudWW14}
Amir Abboud, Virginia~Vassilevska Williams, and Oren Weimann.
\newblock Consequences of faster alignment of sequences.
\newblock In Javier Esparza, Pierre Fraigniaud, Thore Husfeldt, and Elias
  Koutsoupias, editors, {\em Automata, Languages, and Programming - 41st
  International Colloquium, {ICALP} 2014, Copenhagen, Denmark, July 8-11, 2014,
  Proceedings, Part {I}}, volume 8572 of {\em Lecture Notes in Computer
  Science}, pages 39--51. Springer, 2014.

\bibitem[BD21]{BringmannD21}
Karl Bringmann and Debarati Das.
\newblock A linear-time n\({}^{\mbox{0.4}}\)-approximation for longest common
  subsequence.
\newblock In Nikhil Bansal, Emanuela Merelli, and James Worrell, editors, {\em
  48th International Colloquium on Automata, Languages, and Programming,
  {ICALP} 2021, July 12-16, 2021, Glasgow, Scotland (Virtual Conference)},
  volume 198 of {\em LIPIcs}, pages 39:1--39:20. Schloss Dagstuhl -
  Leibniz-Zentrum f{\"{u}}r Informatik, 2021.

\bibitem[BGH17]{BukhGH17}
Boris Bukh, Venkatesan Guruswami, and Johan H{\aa}stad.
\newblock An improved bound on the fraction of correctable deletions.
\newblock {\em {IEEE} Trans. Inf. Theory}, 63(1):93--103, 2017.

\bibitem[BI15]{BackursI15}
Arturs Backurs and Piotr Indyk.
\newblock Edit distance cannot be computed in strongly subquadratic time
  (unless {SETH} is false).
\newblock In Rocco~A. Servedio and Ronitt Rubinfeld, editors, {\em Proceedings
  of the Forty-Seventh Annual {ACM} on Symposium on Theory of Computing, {STOC}
  2015, Portland, OR, USA, June 14-17, 2015}, pages 51--58. {ACM}, 2015.

\bibitem[BJKK04]{BarYossefJKK04}
Ziv Bar{-}Yossef, T.~S. Jayram, Robert Krauthgamer, and Ravi Kumar.
\newblock Approximating edit distance efficiently.
\newblock In {\em 45th Symposium on Foundations of Computer Science {(FOCS}
  2004), 17-19 October 2004, Rome, Italy, Proceedings}, pages 550--559. {IEEE}
  Computer Society, 2004.

\bibitem[BK15]{BringmannK15}
Karl Bringmann and Marvin K{\"{u}}nnemann.
\newblock Quadratic conditional lower bounds for string problems and dynamic
  time warping.
\newblock In Venkatesan Guruswami, editor, {\em {IEEE} 56th Annual Symposium on
  Foundations of Computer Science, {FOCS} 2015, Berkeley, CA, USA, 17-20
  October, 2015}, pages 79--97. {IEEE} Computer Society, 2015.

\bibitem[BR20]{BrakensiekR20}
Joshua Brakensiek and Aviad Rubinstein.
\newblock Constant-factor approximation of near-linear edit distance in
  near-linear time.
\newblock In Konstantin Makarychev, Yury Makarychev, Madhur Tulsiani, Gautam
  Kamath, and Julia Chuzhoy, editors, {\em Proccedings of the 52nd Annual {ACM}
  {SIGACT} Symposium on Theory of Computing, {STOC} 2020, Chicago, IL, USA,
  June 22-26, 2020}, pages 685--698. {ACM}, 2020.

\bibitem[CDG{\etalchar{+}}20]{CDGKS20}
Diptarka Chakraborty, Debarati Das, Elazar Goldenberg, Michal Kouck{\'{y}}, and
  Michael~E. Saks.
\newblock Approximating edit distance within constant factor in truly
  sub-quadratic time.
\newblock {\em J. {ACM}}, 67(6):36:1--36:22, 2020.

\bibitem[CGL{\etalchar{+}}19]{ChenGLRR19}
Lijie Chen, Shafi Goldwasser, Kaifeng Lyu, Guy~N Rothblum, and Aviad
  Rubinstein.
\newblock Fine-grained complexity meets ip= pspace.
\newblock In {\em Proceedings of the Thirtieth Annual ACM-SIAM Symposium on
  Discrete Algorithms}, pages 1--20. SIAM, 2019.

\bibitem[CKK72]{Knuth72}
Va{\v{s}}ek Chv{\'a}tal, David~A Klarner, and Donald~Ervin Knuth.
\newblock {\em Selected combinatorial research problems}.
\newblock Computer Science Department, Stanford University, 1972.

\bibitem[GHL22]{GuruswamiHL22}
Venkatesan Guruswami, Xiaoyu He, and Ray Li.
\newblock The zero-rate threshold for adversarial bit-deletions is less than
  1/2.
\newblock In {\em 2021 IEEE 62nd Annual Symposium on Foundations of Computer
  Science (FOCS)}, pages 727--738. IEEE, 2022.

\bibitem[GRS20]{GoldenbergRS20}
Elazar Goldenberg, Aviad Rubinstein, and Barna Saha.
\newblock Does preprocessing help in fast sequence comparisons?
\newblock In Konstantin Makarychev, Yury Makarychev, Madhur Tulsiani, Gautam
  Kamath, and Julia Chuzhoy, editors, {\em Proccedings of the 52nd Annual {ACM}
  {SIGACT} Symposium on Theory of Computing, {STOC} 2020, Chicago, IL, USA,
  June 22-26, 2020}, pages 657--670. {ACM}, 2020.

\bibitem[HSSS19]{HajiaghayiSSS19}
MohammadTaghi Hajiaghayi, Masoud Seddighin, Saeed Seddighin, and Xiaorui Sun.
\newblock Approximating {LCS} in linear time: Beating the {\(\surd\)}n barrier.
\newblock In Timothy~M. Chan, editor, {\em Proceedings of the Thirtieth Annual
  {ACM-SIAM} Symposium on Discrete Algorithms, {SODA} 2019, San Diego,
  California, USA, January 6-9, 2019}, pages 1181--1200. {SIAM}, 2019.

\bibitem[KS20]{KouckyS20}
Michal Kouck{\'{y}} and Michael~E. Saks.
\newblock Constant factor approximations to edit distance on far input pairs in
  nearly linear time.
\newblock In Konstantin Makarychev, Yury Makarychev, Madhur Tulsiani, Gautam
  Kamath, and Julia Chuzhoy, editors, {\em Proccedings of the 52nd Annual {ACM}
  {SIGACT} Symposium on Theory of Computing, {STOC} 2020, Chicago, IL, USA,
  June 22-26, 2020}, pages 699--712. {ACM}, 2020.

\bibitem[Lev66]{Levenshtein66}
Vladmir~I. Levenshtein.
\newblock Binary codes capable of correcting deletions, insertions and
  reversals.
\newblock {\em Soviet Physics Dokl. (English Translation)}, 10(8):707--710,
  1966.

\bibitem[MP80]{MasekP80}
William~J. Masek and Mike Paterson.
\newblock A faster algorithm computing string edit distances.
\newblock {\em J. Comput. Syst. Sci.}, 20(1):18--31, 1980.

\bibitem[Nos21]{Nosatzki21}
Negev~Shekel Nosatzki.
\newblock Approximating the longest common subsequence problem within a
  sub-polynomial factor in linear time.
\newblock {\em CoRR}, abs/2112.08454, 2021.

\bibitem[RS20]{RubinsteinS20}
Aviad Rubinstein and Zhao Song.
\newblock Reducing approximate longest common subsequence to approximate edit
  distance.
\newblock In Shuchi Chawla, editor, {\em Proceedings of the 2020 {ACM-SIAM}
  Symposium on Discrete Algorithms, {SODA} 2020, Salt Lake City, UT, USA,
  January 5-8, 2020}, pages 1591--1600. {SIAM}, 2020.

\bibitem[RSSS19]{RubinsteinSSS19}
Aviad Rubinstein, Saeed Seddighin, Zhao Song, and Xiaorui Sun.
\newblock Approximation algorithms for {LCS} and {LIS} with truly improved
  running times.
\newblock In David Zuckerman, editor, {\em 60th {IEEE} Annual Symposium on
  Foundations of Computer Science, {FOCS} 2019, Baltimore, Maryland, USA,
  November 9-12, 2019}, pages 1121--1145. {IEEE} Computer Society, 2019.

\bibitem[Ull67]{Ullman67}
Jeffrey~D. Ullman.
\newblock On the capabilities of codes to correct synchronization errors.
\newblock {\em IEEE Transactions on Information Theory}, 13(1):95--105, 1967.

\end{thebibliography}

%%%%%%%%%%% STRUCT-APPENDIX %%%%%%%%%%%%%%%%%%%%%%%%%

\appendix

\section{Proof of Lemma~\ref{lem:comb-struct}}\label{apx:comb-struct}

Lemma~\ref{lem:comb-struct} is essentially a corollary of the stronger combinatorial structure lemma ~\cite[Lemma 4.1]{GuruswamiHL22}, except that the constant dependences are superior and we make the additional assumption that the lengths involved are all powers of two. For completeness, we include a proof here which is significantly simpler than the proof of \cite[Lemma 4.1]{GuruswamiHL22}.

\begin{lemma*}[Lemma~\ref{lem:comb-struct}, restated]
  For $\eps = 10^{-5}$ and $w$ sufficiently large, at least one of the following two conditions holds for every string $x \in \{0,1\}^w$.
  \begin{enumerate}
  \item There exists $\ell \in [\eps^2 w, w]$ equal to a power of two and an $0.1$-imbalanced interval $I$ in $x$ of length $\ell$.

  \item There exists $\ell \in [1, \eps^2 w)$ equal to a power of two such that the number of $\ell^+$-flags in $x$ is at least $\eps w$, and $x$ contains $(0^{\ell}1^{\ell})^{\eps w/\ell}$ as a subsequence.
  \end{enumerate}
\end{lemma*}

\begin{proof}
  We first reduce to the case that $w$ is a power of two. Indeed, suppose we show the statement for all lengths $w'$ equal to sufficiently large powers of $2$, with a stronger $\eps' = 10^{-4}$ in place of $\eps$. Then, let $w'$ be the largest power of two at most $w$, and let $x'=x_{[w']}$ be the prefix of $x$ of length $w'> w/2$. Applying our assumption to $x'$, the lemma statement holds for $x'$ with stronger $\eps' = 10^{-4}$. If $x'$ falls into the first case of the lemma, then $x'$ contains a $0.1$-imbalanced interval $I$ of length $\ell \in [(\eps')^2w', w']\subseteq [\eps^2 w, w]$, so $x$ must fall into the first case as well. 
  
  Otherwise, there exists $\ell \in [1, (\eps')^2 w')$ equal to a power of two such that the number of $\ell^+$-flags in $x'$ is at least $(\eps')^2 w' \ge \eps w$, and $x'$ contains $(0^{\ell}1^{\ell})^{(\eps')^2 w'/\ell} \supseteq (0^{\ell}1^{\ell})^{\eps^2 w/\ell}$. If $\ell \ge \eps^2 w$, then the existence of an $\ell^+$-flag implies that there is a $0.1$-imbalanced interval of length at least $\ell$ in $x'$, so $x$ again falls into the first case of the lemma. On the other hand, if $\ell < \eps^2 w$ then $x$ falls into the second case of the lemma, as desired.

  Thus, we assume $w$ is a power of two and prove this special case with the stronger constant $\eps = 10^{-4}$. Let $w=2^K$, and for any $0\le k \le K$ and $1\le i \le 2^{K-k}$, define $I_{k,i} \defeq [(i-1)\cdot 2^k + 1, i \cdot 2^k]$ to be an aligned dyadic interval of length $2^k$. Observe that for each $k$, the intervals $I_{k,i}$ form a partition of $[w]$. If $I_{k,i}$ is $0.1$-imbalanced for some $k$ satisfying $2^k \ge \eps^2 w$, case 1 holds and we are done. Thus, we may assume $I_{k,i}$ is $0.1$-balanced whenever $2^k \ge \eps^2 w$. We would like to show that case 2 above always holds.
  
  Call an interval $I$ is {\it sparse} if $d(x_I) < 0.01$, and {\it dense} otherwise. Let $\cS_k$ denote the collection of all maximal sparse dyadic intervals $I_{k,i}$ of length $2^k$, i.e. all sparse dyadic intervals $I_{k,i}$ that are not proper subintervals of other sparse $I_{k',i'}$. Let $\cS = \bigcup_{k=0}^K \cS_k$, so that $\cS$ is the collection of all maximal sparse dyadic intervals in $x$, and the elements of $\cS$ are pairwise disjoint.
  
  Observe that sparse intervals are certainly $0.1$-imbalanced, so by our previous assumption, $\cS_k$ is empty if $2^k\ge \eps w$. On the other hand, we also assumed that $I_{K,1} = [w]$ is $0.1$-balanced, so the number of zeros in $x$ is at least $0.4w$. Every zero-bit in $x$ constitutes a sparse dyadic interval $I_{0,i}$ of length $1$ by itself, and every sparse dyadic interval lies inside some element of $\cS$. Thus, intervals in $\cS$ cover all zero-bits in $x$ and have total length at least $0.4w$.
  
  Let if $I=I_{k,i}$ and $i > 1$, define the {\it predecessor} of $I$ to be $\pred(I) \defeq I_{k,i-1}$.
  
  \begin{claim*}
  If $k\ge 0$, $1 < i \le 2^{K-k}$, $I_{k,i} \in \cS$, $t = 2^{\max(0, k-5)}$, and $\pred(I_{k,i})$ is dense, then the number of $t$-flags in $\pred(I_{k,i})$ is at least $0.01 \cdot |\pred(I_{k,i})|$.
  \end{claim*}
  
  \begin{proof}
    If $k<5$ then the assumption that $I_{k,i}$ is sparse implies that it contains only zeros, so the first one-bit in $\pred(I_{k,i})$ is a $1$-flag, and this is sufficient. Assume now that $k\ge 5$. Observe that since $I_{k,i}$ is sparse, it contains at least $0.99 \cdot 2^{k} > 2^{k-1} > 10(t-1)$ zeros and at most $0.01 \cdot 2^k < 2^{k-6} = t/2$ ones. In particular, the last $t/2 = 2^{k-6}$ ones in $x_{\pred(I_{k,i})}$ (or all of them if there are fewer than $2^{k-6}$) must all be $t$-flags. As $\pred(I_{k,i})$ is dense, we are done.
  \end{proof}
  
  Thus, dense predecessors of elements of $\cS$ contain many flags. In order to make sure these flags are not double-counted, we first pass to a subcollection of $\cS$, defined as follows. Write if $I,J\in \cS$, write $I \prec J$ if both $\pred(I)$ and $\pred(J)$ exist, and $\pred(I) \subset \pred(J)$. Define $\cS'$ to be the subcollection of $\cS$ obtained by removing the (at most one) element of the form $I_{k,0}$ without a predecessor, and then removing all elements non-maximal with respect to $\prec$. Observe that if two dyadic intervals satisfy $I \precneq J$, then $I\subseteq \pred(J)$, so for any dyadic interval $J$, the total length of all elements $I$ of $\cS$ satisfying $I\precneq J$ is at most $|J|$. By passing to $\cS'$, we deleted at most half of the total length in $\cS$, plus possibly one interval with no predecessor, which has length at most $\eps^2 w$. Thus,
  \[
  \sum_{I\in \cS'} |I| \ge \frac{1}{2}\sum_{I\in \cS} |I| - \eps^2 w \ge 0.1w.
  \]
  
  Writing $\cS'_{\ge k}$ for the collection of all intervals in $\cS'$ with length at least $2^k$, we pick $k_0$ to be the largest $0 \le k_0 \le K$ for which
  \[
  \sum_{I\in \cS'_{\ge k_0}} |I| \ge 0.01w.
  \]
  Our choice of $\ell$ is $\ell\defeq 2^{\max(0, k_0-5)}$. Note that $\ell < \eps^2 w$ because $\cS_k$ is empty when $2^k \ge \eps^2 w$. We separately prove each of the two required hypotheses.
  
  \begin{claim*}
  For $\ell = 2^{\max(0, k_0-5)}$, the number of $\ell^+$-flags in $x$ is at least $\eps w$.
  \end{claim*}
  \begin{proof}
    For any two dyadic intervals $I$, $J$, either $I \prec J$ or $I \cap J = \emptyset$. Thus, $\{\pred(I)| I \in \cS'_{\ge k}\}$ is a collection of pairwise-disjoint intervals with total length at least $0.01w$. By the previous claim, the number of $\ell^+$-flags in one of these intervals $\pred(I)$ is at least $0.01|\pred(I)| = 0.01|I|$, and so in total the number of $\ell^+$-flags in $x$ is at least $10^{-4}w$, as desired.
  \end{proof} 
  
  It remains to check that $x$ contains $(0^{\ell}1^{\ell})^{\eps w/\ell}$.
  
  \begin{claim*}
  For $\ell = 2^{\max(0, k_0-5)}$, $x$ contains $(0^{\ell}1^{\ell})^{\eps w/\ell}$ as a subsequence.
  \end{claim*}
  \begin{proof}
    Let $k= k_0+1$. By the maximality of $k_0$, we have $\sum_{I\in \cS'_{\ge k}} |I| < 0.01w$. Let $\cS_{\ge k}$ denote the collection of maximal sparse dyadic intervals of length at least $2^k$. Reversing the analysis which led to a lower bound on the total length of $\cS'$, we obtain 
    \[
    \sum_{I\in \cS_{\ge k}} |I| \le 2\sum_{I\in \cS'_{\ge k}} |I| + \eps w \le 0.1w.
    \]
    Since all sparse dyadic intervals of length $2^k$ lie inside some element of $\cS_{\ge k}$, we see that in total at most $0.1 \cdot 2^{K-k}$ of the dyadic intervals $I_{k,i}$ are sparse.
    
    On the other hand, at most $0.7 \cdot 2^{K-k}$ of them have density greater than $0.99$, since otherwise these very dense intervals alone account for at least $0.68w$ ones, making the entire interval $[w]$ $0.1$-imbalanced, which is a contradiction. In sum, out of $2^{K-k}$ total intervals $I_{k,i}$, at most $0.1 \cdot 2^{K-k}$ have density less than $0.01$, and at most $0.7 \cdot 2^{K-k}$ have density greater than $0.99$, leaving at least $0.2\cdot 2^{K-k}$ that must each contain $0.01 \cdot 2^k$ zeros and $0.01 \cdot 2^k$ ones. Passing to only these subintervals, we conclude that $x$ contains a subsequence of the form $x' = x_1 x_2 \cdots x_{0.2\cdot 2^{K-k}}$ where each $x_i$ contains $0.01 \cdot 2^k$ zeros and $0.01 \cdot 2^k$ ones. A string of the form $(1^\ell 0^\ell)^a$ can be found as a subsequence of $x'$ by taking ones from the first $\ceil{100\ell/2^k}$ $x_i$'s, then zeros from the next $\ceil{100\ell/2^k}$, and so on. Since $\ell \ge 2^{k-6}$, we can pick 
    \[
    a \ge \frac{ 0.2 \cdot 2^{K-k}}{2\ceil{100\ell/2^k}} \ge 10^{-4}w/\ell,
    \]
    as desired.
  \end{proof}

\noindent  Combining the above two claims proves that if case 1 of the lemma does not hold, then case 2 does for $\ell = 2^{\max(0, k_0-5)}$.
\end{proof}

\end{document}